\documentclass[11pt, letterpaper]{article}
\usepackage{geometry}
\usepackage[T1]{fontenc}
\usepackage{graphicx}
\usepackage{authblk}
\usepackage{algorithm,algpseudocode}
\usepackage{xspace}
\usepackage[utf8]{inputenc}
\usepackage{cite}
\usepackage{amsmath, amssymb, amsfonts, amsthm}
\usepackage{enumerate}
\usepackage[hypertexnames=false]{hyperref}
\usepackage{paralist}
\usepackage{gensymb}
\usepackage{rotating}
\usepackage{mdframed}
\usepackage{multirow}
\usepackage{verbatim} 
\usepackage{mleftright}
\usepackage{todonotes}
\usepackage{tablefootnote}
\usepackage[capitalize]{cleveref}

\newtheorem{theorem}{Theorem}[section]
\newtheorem{lemma}[theorem]{Lemma}
\newtheorem{corollary}[theorem]{Corollary}

\newtheorem{observation}[theorem]{Observation}

\theoremstyle{definition}
\newtheorem{definition}[theorem]{Definition}

\theoremstyle{remark}
\newtheorem{remark}[theorem]{Remark}

\newenvironment{reminder}[1]{\medskip
	\noindent {\bf Reminder of #1.}\em}{
	\smallskip}

\algtext*{EndWhile}
\algtext*{EndIf}
\algtext*{EndFor}
\algtext*{EndFunction}

\algnewcommand{\IfThen}[2]
{\State \algorithmicif\ #1\ \algorithmicthen\ #2}

\makeatletter
\def\moverlay{\mathpalette\mov@rlay}
\def\mov@rlay#1#2{\leavevmode\vtop{%
	\baselineskip\z@skip \lineskiplimit-\maxdimen
	\ialign{\hfil$\m@th#1##$\hfil\cr#2\crcr}}}
\newcommand{\charfusion}[3][\mathord]{
	#1{\ifx#1\mathop\vphantom{#2}\fi
		\mathpalette\mov@rlay{#2\cr#3}
	}
	\ifx#1\mathop\expandafter\displaylimits\fi}
\makeatother

\def\poly{\mathrm{poly}}
\def\Trunk{\mathrm{Trunk}}
\def\Leaf{\mathrm{Leaf}}
\def\pdeg{\mathrm{pdeg}}
\def\Hi{\mathrm{Hi}}
\def\ET{\mathrm{ET}}

\def\dep{\mathrm{dep}}
\def\avoid{\mathrm{avoid}}
\def\DecTree{\mathrm{DecTree}}
\def\seg{\mathrm{seg}}
\def\parent{\mathrm{parent}}
\def\children{\mathrm{children}}
\def\root{\mathrm{root}}
\def\Pans{{P_{\rm ans}}}
\def\alphanext{\alpha_{\rm next}}
\def\afinal{\alpha_{\rm last}}
\def\istar{{i_\star}}

\def\caI{\mathcal{I}}

\def\caO{\mathcal{O}}
\def\caP{\mathcal{P}}
\def\caS{\mathcal{S}}
\def\caT{\mathcal{T}}
\def\caU{\mathcal{U}}

\def\ie{i.e.\@\xspace}
\def\eg{e.g.\@\xspace}
\def\etal{et al.\@\xspace}

\renewcommand{\varnothing}{\emptyset}

\newcommand{\SpaceQueryStretch}[3]{
	\begin{array}{rl}
		\text{space complexity} & {#1,}\\
		\text{query complexity} & {#2,}\\
		\text{stretch} & {#3}
	\end{array}
}

\begin{document}
	
	\title{Approximate Distance Oracles Subject to Multiple Vertex Failures\thanks{This work has been supported in part by the Zhongguancun Haihua Institute for Frontier Information Technology.}}
	
	\author[1]{Ran Duan \thanks{duanran@mail.tsinghua.edu.cn.}}
	\author[1]{Yong Gu \thanks{guyong12@mails.tsinghua.edu.cn.}}
	\author[1]{Hanlin Ren \thanks{rhl16@mails.tsinghua.edu.cn.}}
	
	\affil[1]{Institute for Interdisciplinary Information Sciences, Tsinghua University}
	
	\maketitle
	
	\begin{abstract}
		Given an undirected graph $G=(V,E)$ of $n$ vertices and $m$ edges with weights in $[1,W]$, we construct vertex sensitive distance oracles (VSDO), which are data structures that preprocess the graph, and answer the following kind of queries: Given a source vertex $u$, a target vertex $v$, and a batch of $d$ failed vertices $D$, output (an approximation of) the distance between $u$ and $v$ in $G-D$ (that is, the graph $G$ with vertices in $D$ removed). An oracle has stretch $\alpha$ if it always holds that $\delta_{G-D}(u,v)\le\tilde{\delta}(u,v,D)\le\alpha\cdot\delta_{G-D}(u,v)$, where $\delta_{G-D}(u,v)$ is the actual distance between $u$ and $v$ in $G-D$, and $\tilde{\delta}(u,v,D)$ is the distance reported by the oracle.
		
		In this paper we construct efficient VSDOs for any number $d$ of failures. For any constant $c\geq 1$, we propose two oracles:\begin{itemize}
			\item The first oracle has size $n^{2+1/c}(\log n/\epsilon)^{O(d)}\cdot \log W$, answers a query in $\poly(\log n,d^c, \\ \log\log W,\epsilon^{-1})$ time, and has stretch $1+\epsilon$, for any constant $\epsilon>0$.
			\item The second oracle has size $n^{2+1/c}\poly(\log (nW),d)$, answers a query in $\poly(\log n,d^c,\\ \log\log W)$ time, and has stretch $\poly(\log n,d)$.
		\end{itemize}
		Both of these oracles can be preprocessed in time polynomial in their space complexity. These results are the first approximate distance oracles of poly-logarithmic query time for any constant number of vertex failures in general undirected graphs. Previously there are $(1+\epsilon)$-approximate $d$-\emph{edge} sensitive distance oracles [Chechik \etal 2017] answering distance queries when $d$ edges fail, which have size $O(n^2(\log n/\epsilon)^d\cdot d\log W)$ and query time $\poly(\log n, d, \log\log W)$.
	\end{abstract}

    \thispagestyle{empty}
	\clearpage
	\pagestyle{plain}
	\pagenumbering{arabic}
	
	\section{Introduction}\label{sec:intro}
	Real-life networks are prone to failures. Usually, there can be several failed nodes or links, but the graph topology will not deviate too much from the underlying failure-free graph. A typical problem is to find the shortest path between two nodes in a network that avoids a specific set of failed nodes or links. This motivates the \emph{$d$-failure} model, in which we should preprocess a graph, such that upon a small number ($d$) of failures, we can ``recover'' from these failures quickly.
	
	In their pioneering work, Demetrescu and Thorup \cite{DT02} designed a data structure that can maintain all-pairs shortest paths under one edge failure. In other words, for each triple $(u,v,f)$ where $u,v$ are vertices and $f$ is a failed edge, the data structure can output the length of the shortest path from $u$ to $v$ that does not go through $f$, in $O(\log n)$ query time. A subsequent work \cite{DTCR08} extends the structure to also handle one vertex failure, and improves the query time to $O(1)$. The one-failure case is studied extensively in literature \cite{CR02, BK08, BK09, DZ17, WY13, GW12, BK13, BCGLPP18, GS18, ChechikC20, Ren20}.
	
	People also tried to find structures handling multiple failures. For undirected graphs, we can answer connectivity queries under $d$ edge failures\footnote{We can also use dynamic connectivity structures with poly-logarithmic worst case update time~\cite{KKM13,GKKT15,Wang15} to handle $d$ edge failures.} \cite{PT07, DP10, DP17} and $d$ vertex failures \cite{DP10, DP17} in $\poly(d,\log n)$ time. Chechik \etal~\cite{CLPR12} designed a data structure that maintains $O(d)$-approximate shortest paths under $d$ edge failures in an undirected graph, and Bil\`o \etal~\cite{BGLP16} improved the approximation ratio to $2d+1$. For any $\epsilon>0$, Chechik \etal~\cite{CCFK17} designed a data structure that $(1+\epsilon)$-approximates shortest paths under $d$ edge failures in an undirected graph, with space complexity $O(n^2(\log n/\epsilon)^d\cdot d\log W)$ and query time $\poly(\log n, d, \log\log W)$, where $W$ is the ratio of the largest edge weight to the smallest edge weight. More related work can be found in \cref{sec:related-work}.
	
	However, despite much effort, it was not known if one can maintain (approximate) shortest paths under multiple \emph{vertex} failures. This problem was addressed as an open problem in \cite{BK13, CLPR12, CCFK17}, and also in Chechik's PhD thesis \cite{Chechik_PhD}.
	
	In this paper we build efficient data structures that answer approximate distance queries under multiple vertex failures for general undirected graphs, answering the above question in the affirmative. A \emph{vertex-sensitive distance oracle} (VSDO) for a weighted undirected graph $G=(V,E)$ is a data structure that given a set of failed vertices $D\subseteq V$ and $u,v\in V\setminus D$, outputs (an estimate of) the length of the shortest path from $u$ to $v$ that avoids all vertices in $D$. We assume a known upper bound $d$ on the number of failures, \ie for any query $(u,v,D)$, we always have $|D|\le d$. We will be concerned with the following parameters of a VSDO:\begin{itemize}
		\item Space complexity, \ie the amount of space that the data structure occupies.
		\item Query time, \ie the time needed to answer one query $(u,v,D)$.
		\item Approximation ratio, a.k.a.~stretch: A VSDO has stretch $\alpha$ if it always holds that  $\delta_{G-D}(u,v)\le\tilde{\delta}(u,v,D)\le\alpha\cdot\delta_{G-D}(u,v)$, where $\delta_{G-D}(u,v)$ is the actual distance between $u$ and $v$ in $G-D$ (\ie $G$ with $D$ disabled), and $\tilde{\delta}(u,v,D)$ is the output of the VSDO. 
	\end{itemize}

	We will not be particularly interested in the preprocessing time of VSDOs; nevertheless, all VSDOs in this paper can be preprocessed in time polynomial in their space complexity.

	In this paper, $n$ and $m$ denote the number of vertices and edges respectively. Let $W$ be the ratio of the largest edge weight to the smallest edge weight. W.l.o.g.~we can assume that edge weights are real numbers in $[1,W]$.

	\subsection{Our Results}\label{sec:our-results}
	We provide the first constructions of approximate VSDOs for general undirected graphs with poly-logarithmic query time. Our main results are as follows:\footnote{$\tilde{O}$ hides $\poly(\log n)$ factors.}
	
	\begin{theorem}[main]
		For any constants $c\ge 1$ and $\epsilon>0$, we can construct VSDOs for undirected graphs with:\label{thm:main}
		\begin{enumerate}[(a)]
			\item space complexity $n^{2+1/c}\log W\cdot (\epsilon^{-1}\log n)^{O(d)}$, query time $\tilde{O}(d^{2c+6}\epsilon^{-1}\log\log W)$ and stretch $1+\epsilon$; \label{item:epsilon_result}
			\item space complexity $\tilde{O}(n^{2+1/c}d^3\log(nW))$, query time $\tilde{O}(d^{2c+9}\log\log(nW))$ and stretch $O(d^{c+2}\log^6n)$. \label{item:polynomial_result}
		\end{enumerate}
		\vspace{-1em}
		Each oracle can be preprocessed in time polynomial in their space complexity.\footnote{See \cref{fig:our-results} in \cref{apd:tables} for precise time bounds.} Our constructions also allow an actual approximate shortest path to be retrieved in an additional time of $O(\ell)$, where $\ell$ is the number of edges in the reported path. 
	\end{theorem}

	Using existing structures, we need either $n^{\Omega(d)}$ space or $\Omega(n)$ query time.\footnote{We can use the $d$-fault tolerant spanner \cite{CLPR09, DK11, BDPW18, BP19} with the brute-force query algorithm, build $n^{d-2}$ two-failure distance oracles~\cite{DP09}, use the dynamic shortest path algorithms~\cite{BrandN19}, or use the oracle \cite{BS19} which also works for directed graphs. But none of these solutions provide both $n^{o(d)}$ space and $o(n)$ query time.} Thus our results are the first of its kind.

	\subsection{A Brief Overview}\label{sec:brief-overview}

	In this section, we briefly introduce the ideas needed to construct the desired VSDOs.

	\paragraph{The edge-sensitive distance oracle of \cite{CCFK17}.} Our first VSDO depends on \cite{CCFK17} which handles $d$ \emph{edge} failures. Therefore we briefly describe their oracle first. It may be helpful to think of their query algorithm as a recursive one. 
	
	Given $u,v\in V$ and a set $D$ of $d$ edge failures, let $\Pans$ be the shortest $u$-$v$ path in $G-D$, which we are searching for. The oracle first partitions the shortest path $P$ from $u$ to $v$ in $G$ (which may go through failures) into $\tilde{O}(\epsilon^{-1}\log W)$ short segments. Consider a segment $X$ that contains some failed edges. If $\Pans$ does not go through $X$, then we can ``preprocess'' the graph $G-X$ and search for $\Pans$ in $G-X$. Otherwise, if $\Pans$ goes through some vertex $x\in X$, then we can pick an \emph{arbitrary} vertex $w\in X$ such that there are no failed edges between $x$ and $w$, and \emph{pretend} that $\Pans$ passes through $w$. That is, we recursively find the shortest paths in $G-D$ from $u$ to $w$ and from $w$ to $v$ and concatenate them. It is easy to see that this brings an additive error of at most $2|X|$ to our solution, where $|X|$ is the length of $X$.
	
	Thus, we want to find a small set of intermediate vertices, which we denote as $H$, with the following property: For every vertex $x$ and failure $f$, if $x$ has distance at most $|X|$ to $f$, then there is some $w\in H$ such that $x$ also has distance at most $|X|$ to $w$ in $G-D$. As it turns out that the query time is polynomial in $|H|$, the size of $H$ should be small.
	
	There is a natural choice of $H$: we simply let it be the set of vertices incident to some failed edges. It is easy to see that $|H|\le 2d$, thus the query algorithm runs in time $\poly(d)$. The above property is also true: given any vertex $x$ and a nearby failure $f$, we can walk along the path from $x$ to $f$ until we meet a failed edge, then the vertex $w$ we stop at is both in $H$ and close to $x$. We can control the total additive error (\ie~the sum of $2|X|$'s over the ``recursion'') to be at most $\epsilon\cdot |\Pans|$, by partitioning each path into \emph{sufficiently short} segments.
	
	Note that, for the sake of intuition, we have omitted some important details, such as how to ``preprocess'' $G-X$ (by a decision tree structure) and how to implement the query algorithm (non-recursively).
	
	\paragraph{The ``high-degree'' obstacle.} The obvious difficulty of handling vertex failures is the presence of failed vertices with very high degrees. If every failed vertex has degree $\le \Delta$, we can simply simulate an edge-failure distance oracle \cite{CLPR12, CCFK17} and delete at most $d\cdot\Delta$ edges from it. Equivalently, we can define the set of intermediate vertices $H$ as those non-failure vertices adjacent to some failure, then $|H|\le d\cdot\Delta$ and we run the above query algorithm. However the techniques of \cite{CLPR12, CCFK17} do not seem to work for high-degree vertex failures. For example, techniques in \cite{CLPR12} only guarantee a stretch of $\ge\Delta$, and techniques in \cite{CCFK17} require $\poly(\Delta)$ query time, therefore both are unsatisfactory when $\Delta=\Omega(n)$. 
	
	By the construction of $(2k-1)$-stretch spanners with $O(n^{1+1/k})$ edges~\cite{Althofer+93}, we can construct a $(2\log n-1)$-stretch spanner with $O(n)$ edges. We note that the query algorithm works even if every failed vertex has a small degree in the spanner (rather than in the whole graph): We can define $H$ to be the set of vertices adjacent to some failed vertex \emph{in the spanner}. If $\Pans$ goes through some vertex $x$ that has distance $|X|$ to a failed vertex $f$, the distance between $x$ and $f$ in the spanner is $O(|X|\log n)$, and there must be some $w\in H$ that has distance $O(|X|\log n)$ to $x$ in $G-D$. By partitioning the paths into shorter segments, we can still control the additive error, \ie~the sum of $O(|X|\log n)$ over the ``recursion'', to be less than $\epsilon\cdot |\Pans|$.
	
	\paragraph{High-degree hierarchy: A first attempt.} Given the ``high-degree'' obstacle, it is natural to see whether the ``high-degree hierarchy'' of \cite{DP10} may help us. Plugging the spanners\footnote{The reason that we need to plug in a spanner, rather than the original graph, is that we can only plug in a \emph{sparse} graph into the high-degree hierarchy.} into the hierarchy of \cite{DP10}, we obtain a structure as follows. The vertices are partitioned into $p=O(\log n)$ levels; let $U_i$ be the set of vertices with level $\ge i$. So we have a sequence of vertex sets $V=U_1\supseteq U_2\supseteq\dots\supseteq U_p\supseteq U_{p+1}=\varnothing$, and the $i$-th level is the set $U_i\setminus U_{i+1}$.\footnote{In the hierarchy structure of \cref{sec:high-degree-hierarchy}, each $U_{i+1}$ is not necessarily a subset of $U_i$; this issue is not essential, so for simplicity, in the brief overview we will assume each $U_{i+1}$ is indeed a subset of $U_i$.} For every $i$, let $G_i$ be the induced subgraph of $V\setminus U_{i+1}$. We do not have a complete spanner for $G_i$; we can only afford to build a ``subset-spanner'' that preserves the distances in $G_i$, among vertices in $U_i\setminus U_{i+1}$ (instead of $V\setminus U_{i+1}$). The structure guarantees that every failed vertex in the subset-spanner of any level has low degrees.
	
	It is natural to define $H$ as the set of neighbors of failures in the subset-spanners, and $|H|$ will be small. If $\Pans$ goes through some vertex $x$ that has distance $|X|$ to a failed vertex $f$, and \emph{$x$ and $f$ are in the same level}, then we can find an intermediate vertex $w\in H$ that has distance $O(|X|\log n)$ to $x$ in $G-D$, and we are fine. But what if $x$ and $f$ are in different levels? In this case, the $x$-$f$ path may not be preserved by the ``\emph{subset}-spanner'', thus not captured by $H$. In \cite[Section 4]{DP10}, the authors used ad hoc structures to preserve connectivity between different levels; it appears difficult to extend these structures to also handle ($(1+\epsilon)$-approximate) distances.
	
	\paragraph{Our ideas.} It is inconvenient that the spanner at level $i$ only preserves distances inside $U_i\setminus U_{i+1}$. Therefore, our first idea is to ``extend'' the spanners to also preserve distances at lower levels: the spanner at level $i$ should preserve distances between any pair of vertices $(x, y)$, where $x\in U_i\setminus U_{i+1}$ and $y\in V\setminus U_{i+1}$. Note that we still only guarantee that every vertex failure has small degrees in the \emph{original} spanners; they may have large degrees in the extended spanners.
	
	We implement the spanners by \emph{tree covers}, and there is a natural way to ``extend'' them. The extended tree cover consists of a collection of trees whose union is a spanner that preserves distances between $U_i\setminus U_{i+1}$ and $V\setminus U_{i+1}$. Moreover, each tree is a shortest path tree rooted in $U_i\setminus U_{i+1}$ (the highest level of $G_i$). See \cref{sec:sr-tree-cover} for more details.
	
	Recall that in the query algorithm, we have a non-failure vertex $x$ that is close to a failure $f$, and we want to find an intermediate vertex $w\in H$ that is close to $x$ in $G-D$. Suppose that $x$ is at a higher level than $f$. If we walk from $f$ (at a lower level) to $x$ (at a higher level), it seems that our first step should go to the parent of $f$ in some tree. Actually, this intuition can be rigorously proved! See the proof of \cref{lemma:ball_VH}. Therefore, if $H$ consists of the neighbors of every failure (in the original spanners) and the parents of every failure in each tree (in the extended tree covers), then we can deal with every $(x, f)$ such that the level of $x$ is at least that of $f$. Every failure is only in $\tilde{O}(1)$ trees, thus $|H|$ is indeed small.
	
	We need to adapt the query algorithm to ensure that $f$ never has a higher level than $x$. Let $P$ be the shortest $u$-$v$ path in the original graph, and we partition $P$ into short segments. Consider a segment $X$ that contains failures, and let $i$ be the highest level of any failure in $X$. If $\Pans$ does not contain any vertex in $X$ \emph{with level at least $i$}, then we can ``preprocess'' the graph $G-(X\cap U_i)$ and search for $\Pans$ in this subgraph. Otherwise $\Pans$ goes through some $x\in (X\cap U_i)$, and by definition, the level of $x$ cannot be smaller than the level of any failure in $X$. Therefore, we can find some intermediate vertex $w\in H$ close to $x$, ``pretend'' that $\Pans$ goes through $w$, and continue.
	
	The above discussion implies a data structure with space complexity roughly $n^3$. To reduce the space complexity by a factor of $n^{1-o(1)}$, we prove a structural theorem (\cref{thm:decomposable}) for shortest paths under vertex failures, which allows us to compress such paths. (The corresponding theorem \cite[Theorem 3.1]{CCFK17} does not hold for vertex failures.) Curiously, the proof of this theorem also relies on \cref{lemma:ball_VH}.

	\paragraph{On oracle (\ref{item:polynomial_result}).} Although oracle (\ref{item:polynomial_result}) has a larger stretch compared to oracle (\ref{item:epsilon_result}), we think it is also of interest, since it is the first oracle that handles $\omega(\log n)$ failures in polynomial space and $\poly(\log n)$ query time, within a reasonable stretch.\footnote{It seems that even $O(\sqrt{n})$ stretch was open before this result.} Note that setting $\epsilon=\omega(1)$ (\eg $\epsilon=\log n$) in oracle (\ref{item:epsilon_result}) does \emph{not} improve its space complexity to $n^2\log^{o(d)} n$, so oracle (\ref{item:polynomial_result}) is \emph{not} a direct corollary of oracle (\ref{item:epsilon_result}).

	\subsection{More Related Work}\label{sec:related-work}
	\paragraph{Sensitivity oracles.} For the case of two vertex failures, Duan and Pettie \cite{DP09} showed that exact distances in a directed weighted graph can be queried in $O(\log n)$ time, with an oracle of size $O(n^2\log^3 n)$, and Choudhary \cite{Cho16} designed an oracle of $O(n)$ size that handles single source reachability queries in directed graphs in $O(1)$ time.
	
	The general problem of $d$ failures has also received attention on \emph{planar} graphs: Borradaile \etal \cite{BPW12} constructed a data structure that maintains connectivity under $d$ vertex failures, and Charalampopoulos \etal \cite{CMT19} designed a data structure that answers exact distance queries under $d$ vertex failures. 
	
	In a recent breakthrough, van den Brand and Saranurak \cite{BS19} gave an oracle that handles an arbitrary number $d$ of edge failures in \emph{directed} graphs. Their oracle can answer reachability queries in $O(d^{\omega})$ time, and exact distance queries in $n^{2-\Omega(1)}$ time (for small integer weights), where $\omega<2.3728639$ is the matrix-multiplication exponent \cite{CW90, Sto10, Wil12, LeGall}. 
	
	We summarize the sensitivity connectivity/distance oracles in \cref{fig:other_results} of \cref{apd:tables}.
	
	\paragraph{Fault-tolerant structures.} A related concept is \emph{fault-tolerant (FT) spanners}: a subgraph $G'$ of $G$ is a $d$-FT spanner if, after removing any $d$ vertices, the remaining parts of $G'$ is a spanner of the remaining parts of $G$. It might be \textit{a priori} surprising that sparse FT spanners exist, but Chechik \etal \cite{CLPR09} gave the first construction of $d$-FT $(2k-1)$-spanners with $O(d^2k^{d+1}\cdot n^{1+1/k}\log^{1-1/k}n)$ edges. Subsequent papers \cite{DK11, BDPW18, BP19} improved the number of edges to $O(n^{1+1/k}d^{1-1/k})$, which is optimal assuming the girth conjecture of Erd\H{o}s \cite{Erd64}.
	
	Besides FT spanners, there are many other kinds of fault-tolerant structures, \eg \cite{BGLP14, Par15, PP15, PP16, BGLP16, BGPW17, Par17, PP18}. We refer the reader to the excellent survey of \cite{Par16}.
	
	\paragraph{Dynamic shortest path.} There are dynamic all-pairs shortest path structures handling vertex updates. Thorup~\cite{Thorup05} gave a fully dynamic all-pairs shortest paths structure with worst-case update time $\tilde{O}(n^{2.75})$, and Abraham \etal~\cite{ACK17} gave a randomized worst-case update time bound $\tilde{O}(n^{2+2/3})$. Recently, Brand and Nanongkai~\cite{BrandN19} gave a $(1+\epsilon)$-approximate algorithm for maintaining APSP under edge insertions and deletions with worst-case update time $\tilde{O}(n^{1.863}/\epsilon^2)$ for directed graphs. 
	Other fully or partial dynamic shortest path structures include~\cite{ACT14,Bernstein2009,DI2006,HKN2008,HKN2014,Henzinger2014,King1999,RZ2011,RZ2004,Thorup_SWAT,Sankowski2005}.

	\subsection{Notation}\label{subsec:pre}
	In this paper, $\log x=\log_2 x$, and $\ln x=\log_e x$. For a set $S$ and an integer $k$, $|S|$ is the cardinality of $S$, and we denote $\binom{S}{k}=\{S'\subseteq S:|S'|=k\}$, and $\binom{S}{\le k},\binom{S}{\ge k}$ are defined analogously. For two sets $X$ and $Y$, define their Cartesian product as $X\times Y=\{(x,y):x\in X,y\in Y\}$. We use $\circ$ as the concatenation operator for paths or sequences. For paths $P_1,P_2$, if $u$ is the last vertex in $P_1$ and $v$ is the first vertex in $P_2$, then $P_1\circ P_2$ is well-defined if $u=v$ or $(u,v)$ is an edge in $G$.
	
	For a graph $H$ and $u,v\in V(H)$, $w_H(u,v)$ denotes the length of the edge between $u$ and $v$ ($w_H(u,v)=+\infty$ if such an edge does not exist), $\delta_H(u,v)$ denotes the length of the shortest path in $H$ from $u$ to $v$ and $\pi_H(u,v)$ denotes the corresponding shortest path. If $S\subseteq V(H)$ is a subset of vertices, then $\delta_H(u,S)=\min\{\delta_H(u,v):v\in S\}$. ($\delta_H(u,\emptyset)=+\infty$.) We omit the subscript $H$ if $H=G$ is the input graph. We define $H[S]$ as the subgraph induced by $S$, and $H-S=H[V(H)\setminus S]$. We use $nW$ as an upper bound of the diameter of any (connected) subgraph of $G$. We assume that the shortest path between every pair of vertices in any subgraph is unique (see Section 3.4 of~\cite{DI04}).
	
	For a path $P$ and $u,v\in P$, define $P[u,v]$ as the portion from $u$ to $v$ in $P$, and sometimes this notation emphasizes the \emph{direction} from $u$ to $v$. Let $(u=x_0,x_1,\dots,x_{\ell-1},x_{\ell}=v)$ denote the path $P[u,v]$, then we define $P(u,v]=P[x_1,v],P[u,v)=P[u,x_{\ell-1}]$ and $P(u,v)=P[x_1,x_{\ell-1}]$. Define $|P|$ as the length of path $P$. For a tree $T$ rooted at $r$ and a vertex $x\in V$, define the depth of $x$, denoted by $\dep_T(x)$, as the (weighted) distance from $x$ to $r$ in $T$.
	
	In this paper, $D$ denotes the set of $\le d$ failed vertices. For convenience, we always assume $n\ge 3$ and $d\ge 2$.
	
	Note that we also define some more notations at the end of \cref{sec:high-degree-hierarchy}, which is relevant to the ``high-degree hierarchy''. \cref{fig:def} in \cref{apd:tables} summarizes some nonstandard notation in this paper.
	
	\section{Source-Restricted Tree Covers in High-Degree Hierarchy}\label{sec:treecover}
	Our VSDO is based on a variant of the \emph{high-degree hierarchy} of \cite{DP10}, which we equip with the \emph{source-restricted tree covers} of \cite{TZ05, RTZ05} to approximately preserve distances.

	\subsection{Source-Restricted Tree Covers}\label{sec:sr-tree-cover}
	Let $G=(V,E)$ be an undirected graph. A \emph{tree cover} of $G$ is, informally, a set of trees such that every vertex $v\in V$ is in a small number of trees, and for every two vertices $u,v\in V$, there is a tree that approximately preserves their distance $\delta(u,v)$. In this paper, we relax the second condition, requiring it to hold only for every $u\in S,v\in V$, where $S$ is some subset of $V$. Following terminologies of \cite{RTZ05}, we call such tree covers \emph{source-restricted}.
	
	Throughout this paper, $k=\ln n$.\footnote{Our construction works for any parameter $k$, but the complexity is proportional to $kn^{1/k}$, so we minimize it by setting $k=\ln n$.} We define \emph{source-restricted tree cover} as follows.
	
	\begin{definition}\label{def:sr-tree-cover}
		Given $S\subseteq V$, an \emph{$S$-restricted tree cover} is a set of rooted trees $\{T(w):w\in S\}$, such that the following hold.
		\begin{enumerate}[a)]
			\item For every $w\in S$, there is exactly one tree $T(w)$ rooted at $w$, spanning a subset of $V$ (which we denote as $V(T(w))$).
			\item For every $u\in S,v\in V$, there is some $w\in S$ such that $u,v\in V(T(w))$, and $\dep_{T(w)}(u)+\dep_{T(w)}(v)\le (2k-1)\delta(u,v)$.\footnote{We only require that the \emph{distance} between $u$ and $v$ in $T(w)$ approximates $\delta(u,v)$ well in \cref{sec:orac2}. However, we will require that the \emph{sum of depths} of $u$ and $v$ approximates $\delta(u,v)$ well in \cref{sec:orac1}.}\label{item:treecover-b}
			\item Every vertex $v\in V$ is in at most $kn^{1/k}(\ln n+1)\le 2e\ln^2 n$ trees.\label{item:treecover-c}
	\end{enumerate}
	\end{definition}
	
	In \cite{TZ05}, Thorup and Zwick constructed approximate distance oracles, and they noticed that their constructions are also good tree covers. A simple modification of their construction (see \cite{RTZ05}) yields source-restricted tree covers.
	\begin{theorem}\label{thm:sr-tree-cover}
		Given a graph $G=(V,E)$ and $S\subseteq V$, we can compute in deterministic polynomial time an $S$-restricted tree cover $\caT(S)=\{T(w):w\in S\}$ such that for any $u\in S,v\in V$, the vertex $w$ in \cref{def:sr-tree-cover} \ref{item:treecover-b}) can be found in $O(k)$ time.
	\end{theorem}
	
	For completeness, we provide a sketch of the construction in \cref{sec:proof-tree-cover}; we also refer the interested reader to \cite{TZ05, RTZ05} for details of this construction. 

	For $S\subseteq V$, we denote $\caT(S)$ as the $S$-restricted tree cover constructed in \cref{thm:sr-tree-cover}. For $S,R\subseteq V$, we denote $\caT_R(S)$ as the $(S\setminus R)$-restricted tree cover $\caT(S\setminus R)$ in $G-R$.

	For technical reasons (namely, we want the hierarchy structure in \cref{sec:high-degree-hierarchy} to have a reasonable size), we need that the number of ``high-degree'' vertices in $\caT(S)$ is only $o(|S|/d)$, where $d$ is the number of failures. However, here we defined the tree cover $\caT(S)$ to span not only $S$, but maybe some other vertices in $V$.\footnote{This corresponds to the informal description of ``extending'' tree covers in \cref{sec:brief-overview}.} So we can only prove degree bounds of the following form: the number of vertices with high degree w.r.t.~the ``trunk'' parts of the tree cover is $o(|S|/d)$. The precise definitions are as follows.
	
	\begin{definition}
		Consider $S\subseteq V$, $T\in\caT(S)$, $v\in V(T)$. We say $v$ is a \emph{trunk vertex} of $T$ if there are $u,w\in S$ such that $v$ lies on the path from $u$ to $w$ in $T$. The subtree (subgraph) of $T$ induced by trunk vertices of $T$ is denoted as $\Trunk(T)$. The \emph{pseudo-degree} of a vertex $v\in V(T)$, denoted as $\pdeg_T(v)$, is the degree of $v$ in $\Trunk(T)$. If $v$ is not a trunk vertex of $T$, then $\pdeg_T(v)=0$.
		\label{def:trunk}
	\end{definition}

	Note that vertices in $\Trunk(T)$ are not necessarily in $S$. See \cref{fig:pseudodegree} as an example. 
	
	The following property will be useful in \cref{sec:orac2}: for a vertex $v$ that is not in $\Trunk(T)$, its path in $T$ to any vertex in $S$ must go through its parent. (This is because the root of $T$ is always in $S$.)
	
	\begin{figure}
		\begin{minipage}[c]{0.49\linewidth}
		\centering
		\includegraphics[scale=1]{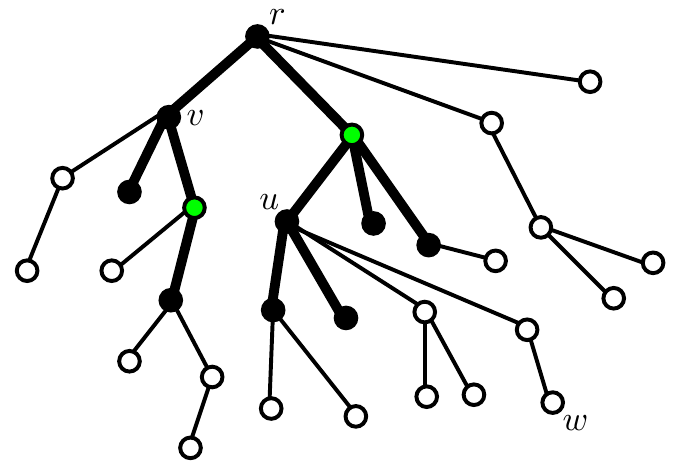}
		\end{minipage}
		\begin{minipage}[c]{0.49\linewidth}
		\begin{tabular}[b]{|c|c|c|}
		\hline
		&degree&pseudo-degree\\
		\hline
		$r$&$4$&$2$\\
		\hline
		$u$&$5$&$3$\\
		\hline
		$v$&$4$&$3$\\
		\hline
		$w$&$1$&$0$\\
		\hline
		\end{tabular}
		\end{minipage}
		\caption{A sample tree in $\caT(S)$. Black vertices are in $S$, green vertices are trunk vertices not in $S$, and bold edges denote the subtree induced by trunk vertices. We also include a table of degrees and pseudo-degrees of some sample vertices.}\label{fig:pseudodegree}
	\end{figure}
	
	Let $s=4e\cdot d^{c+1}\ln^2 n+1$ be a degree threshold, where $c\ge 1$ is any constant. Define $\Hi(\caT(S))$ as the set of vertices in $V$ that has pseudo-degree $>s$ in some tree in $\caT(S)$. We prove our desired upper bound on $|\Hi(\caT(S))|$.
	
	\begin{lemma}
		For any $S\subseteq V$, $|\Hi(\caT(S))|\le \frac{|S|}{2d^{c+1}}$.\label{lemma:high-vertices}
	\end{lemma}
	\begin{proof}
		For a tree $T$, let $\Leaf(T)$ be the set of leaves of $T$. Then there are at most $\left\lfloor\frac{|\Leaf(T)|-2}{s-1}\right\rfloor$ vertices in $T$ that has degree $>s$ \cite[Lemma 3.1]{DP10}. For any $T\in\caT(S)$, $\Leaf(\Trunk(T))\subseteq S$ by definition, thus\begin{equation*}
		\sum_{T\in\caT(S)}\left|\Leaf(\Trunk(T))\right|\le\sum_{v\in S}\left|\{T\in\caT(S):v\in T\}\right|.
		\end{equation*}
		Since every $v\in S$ appears in $\le 2e\ln^2 n$ trees in $\caT(S)$, we have \begin{equation*}
		\sum_{T\in\caT(S)}\left|\Leaf(\Trunk(T))\right|\le|S|\cdot 2e\ln^2 n,
		\end{equation*}
		thus 
		\[
			|\Hi(\caT(S))|\le\sum_{T\in\caT(S)}\left\lfloor\frac{|\Leaf(\Trunk(T))|}{s-1}\right\rfloor
			\le\frac{|S|\cdot 2e\ln^2 n}{s-1}
			=\frac{|S|}{2d^{c+1}}.\qedhere
		\]
	\end{proof}

	\subsection{The High-Degree Hierarchy}\label{sec:high-degree-hierarchy}
	We use a simplified version of the high-degree hierarchy in \cite{DP10}. Fix a parameter $c\ge 1$, the hierarchy structure is a set of $O(n^{1/c})$ representations of the graph, such that for every set of $d$ failures, we can find some representation in which all failed vertices have low pseudo-degrees in their relevant tree covers.
	
	\begin{definition}
		The \emph{hierarchy tree} is a rooted tree in which every node\footnote{We use ``vertex'' for nodes in the input graph, and ``node'' for nodes in the hierarchy tree.} corresponds to a subset of $V$. The root corresponds to $V$. Each node $U$ ($U\subseteq V$) stores a tree cover $\caT(U)$, and each edge $(U,W)$, where $U$ is the parent of $W$, stores a tree cover $\caT_W(U)$, which is $\caT(U\setminus W)$ in $G-W$. 
		The hierarchy tree is constructed as follows. Let $U$ be any node. If $\Hi(\caT(U))=\varnothing$, then $U$ is a leaf; otherwise let $W_1,W_2,\dots,W_d$ be its children, where
		\begin{align}
		W_1 = &~\Hi(\caT(U)), \label{eq:W1}\\
		W_i = &~ W_{i-1}\cup\Hi(\caT_{W_{i-1}}(U)).  &\text{(for $2\leq i\leq d$)}\nonumber
		\end{align}
		Then we recursively deal with all $W_1,W_2,\dots,W_d$.
		\label{def:hierarchy-tree}
	\end{definition}
	
	There are two main differences compared with the hierarchy structure in \cite{DP10}.\begin{itemize}
		\item We simplified the definition of $W_1$ as in \eqref{eq:W1}. This change is not essential, but we feel that it could make the hierarchy tree easier to understand. As a consequence, a node is not necessarily a subset of its parent, which is different from \cite{DP10} (and \cref{sec:brief-overview}).
		\item More importantly, we store in each node the source-restricted tree covers introduced in \cref{sec:sr-tree-cover}. By contrast, \cite{DP10} only concerns about connectivity, so they used a (somewhat arbitrary) spanning forest instead.
	\end{itemize}
	
	The following lemmas assert that the hierarchy tree ``is small, shallow and effectively represents the graph'' \cite{DP10}, which are crucial for our data structures.

	\begin{lemma}[Hierarchy Size and Depth]
		Consider the hierarchy tree constructed with high-degree threshold $s=4e\cdot d^{c+1}\ln^2 n+1$, then the following hold.
		\begin{enumerate}
			\item The depth $h$ of the hierarchy tree is at most $\lfloor\frac{1}{c}\log_d n\rfloor$,\footnote{In subsequent sections we will write $h$ as a shorthand of $O(\log n/\log d)$ in time/space bounds.} assuming the root has depth $0$.
			\item The number of nodes in the hierarchy tree is at most $O(n^{1/c})$.
		\end{enumerate}\label{lemma:size-and-depth}
	\end{lemma}
	\begin{proof}
		Let $U$ be a node in the hierarchy tree, $W_1,W_2,\dots,W_d$ be its children (if exist). By \cref{lemma:high-vertices}, we have $|W_1|=|\Hi(\caT(U))|\le\frac{|U|}{2d^{c+1}}$ and $|\Hi(\caT_{W_i}(U))|\le\frac{|U|}{2d^{c+1}}$ for any $0 < i \le d$. Since $W_{i+1}=W_i\cup \Hi(\caT_{W_i}(U))$, it follows that $|W_i|\le i|U|/2d^{c+1}$ for each $i$. Therefore $|W_i|\le |U|/2d^c$ for all $0 < i \le d$. Any node at the $k$-th level corresponds to a subset of $V$ with size at most $n/(2d^c)^k$, therefore the depth of the hierarchy tree is at most $h\le \lfloor\log_{2d^c}n\rfloor\le\lfloor\frac{1}{c}\log_d n\rfloor$. There are at most $\sum_{i=0}^hd^i\le 2d^h=O(n^{1/c})$ nodes in the hierarchy tree.
	\end{proof}
	
	\begin{algorithm}[t]
		\caption{Path finding algorithm}\label{alg:path-finding}
		\begin{algorithmic}[1]
			\State {$U_1\gets V$}
			\For {$i\gets 1$ to $h$}
			\State {If $U_i$ is a leaf then set $p\gets i$ and halt}\label{line:alg1-halt1}
			\State {Let $W_1,W_2,\dots,W_d$ be the children of $U_{i}$ and artificially define $W_0=\varnothing$, $W_{d+1}=W_d$.}
			\State {Let $j\in[0,d]$ be minimal such that $D\cap (W_{j+1}\setminus W_j)=\varnothing$}\label{line:j}
			\State {If $j=0$ then set $p\gets i$ and halt}\label{line:alg1-halt2}
			\State {Otherwise $U_{i+1}\gets W_j$}
			\EndFor
		\end{algorithmic}
	\end{algorithm}
	
	\begin{lemma}
		For any set $D$ of at most $d$ failures, \cref{alg:path-finding} finds in $O(hd)$ time a path $U_1(=V), U_2, \dots, U_p$ in the hierarchy tree from root to some node $U_p$, such that for every tree $T\in\bigcup_{1\le i\le p}\caT_{U_{i+1}}(U_i)$ and every $f\in D$, $f$ has pseudo-degree $\le s$ in $T$. (Assume $U_{p+1}=\varnothing$.)\label{lemma:good-path}
	\end{lemma}
	\begin{proof}
		\cref{alg:path-finding} executes at most $O(h)$ iterations since the hierarchy tree has depth at most $h$. For every node $U$ and vertex $v\in V$, we store the first child $W_j$ of $U$ (or none) that $v$ appears in. It is then easy to implement each iteration in $O(d)$ time. When \cref{alg:path-finding} halts at \cref{line:alg1-halt1} or \ref{line:alg1-halt2}, either $U_p$ is a leaf or $D\cap W_1=\emptyset$, where $W_1 = \Hi(\caT(U_p))$ is the first child of $U_p$. Clearly, in both case we have $D\cap\Hi(\caT(U_p))=\emptyset$.
		
		For any $1\le i <p$, since $W_{d+1}=W_d$, \cref{line:j} can always find such $j$. Let $U_{i+1} = W_j$ be the $j$-th child of $U_i$. If $j=d$, then $D\cap (W_{j'+1}\setminus W_{j'})\ne\varnothing$ for $j'=0,\dots,d-1$. Since $|D|\le d$, we have $D\subseteq W_d$, thus $D\cap\Hi(\caT_{W_d}(U_i))=\varnothing$. If $1\le j <d$, then we have $D\cap(W_{j+1}\setminus W_j)=\emptyset$. Since $W_{j+1}=W_j\cup \Hi(\caT_{W_j}(U_i))$ and $\Hi(\caT_{W_j}(U_i))\cap W_j=\varnothing$, we have $D\cap \Hi(\caT_{W_j}(U_i))=\varnothing$. The lemma follows.
	\end{proof}
	
	In \cref{sec:orac2}, \cref{sec:reduce-space} and \cref{sec:orac1}, we always deal with a path $V=U_1,U_2,\dots,U_p$ in the hierarchy tree from the root $V$ to a node $U_p$ (not necessarily a leaf node). Artificially define $U_{p+1}=\varnothing$. In other words:
	\begin{itemize}
		\item We run the preprocessing algorithm for every possible such paths $V=U_1,U_2,\dots,U_p$, and we build a separate data structure for each path. The space complexity is then multiplied by a factor of $O(n^{1/c})$.
		\item In the query algorithm, given $D$, we always begin by identifying a path $V=U_1,U_2,\dots,U_p$ using \cref{lemma:good-path}, then every failed vertex $f\in D$ has pseudo-degree $\le s$ in every $\caT_{U_{i+1}}(U_i)$ ($1\le i\le p$).
	\end{itemize}

	Fix a path $V=U_1,U_2,\dots,U_p$ in the hierarchy tree, and we assume $U_{p+1}=\varnothing$. The following corollary of \cref{thm:sr-tree-cover} will be important for us.
	\begin{corollary}
		Let $x\in U_{\ell}\setminus U_{\ell+1}$ and $y\in V\setminus U_{\ell+1}$. There is a tree $T\in\caT_{U_{\ell+1}}(U_\ell)$ such that 
		\[\dep_T(x)+\dep_T(y)\le (2k-1)\delta_{G-U_{\ell+1}}(x,y).\]
		Moreover, the root of such a tree can be found in $O(k)$ time.\label{cor:T-ell(xy)}
	\end{corollary}

	We will denote this tree $T$ as $T_\ell(x,y)$, and denote the path from $x$ to $y$ in $T$ as $\caP_\ell(x,y)$.
	
	The set of trees is denoted as $\caT=\bigcup_{i=1}^p\caT_{U_{i+1}}(U_i)$. The proof of \cref{lemma:size-and-depth} shows that $|U_i|\le |U_{i-1}|/2d^c$, therefore $|\caT|\le\sum_{i=1}^p|U_i|=O(n)$. In a path $V=U_1,U_2,\dots,U_p$ in the hierarchy tree, since $U_{i+1}$ is not necessarily a subset of $U_i$, we define the \emph{level} of a vertex $v$ as:
	\begin{definition}
	    Fix a path $V=U_1,U_2,\dots,U_p$ in the hierarchy tree, define the \emph{level} of $v$ to be the largest integer $l$ such that $v\in U_l$, denoted as $l(v)$. Define $G_{\ell}$ to be the subgraph of $G$ induced by all vertices with level at most $\ell$.
	\end{definition}
	
	\section{An $(1+\epsilon)$-Stretch Oracle with $n^{3+1/c+o(1)}$ Space}\label{sec:orac2}
	In this section we present an oracle with
	\[\SpaceQueryStretch{n^{3+1/c}\cdot (\epsilon^{-1}\log(nW))^{O(d)}}{\poly(\log(nW),\epsilon^{-1},d)}{1+\epsilon,}\]
	for any $\epsilon>0$. In this paper (except \cref{sec:orac1} and \cref{sec:arbitrary-to-bounded}) we may assume $d=o\mleft(\frac{\log n}{\log\log n}\mright)$, $W=\poly(n)$, so we can simplify the notation for space complexity to $n^{3+1/c+o(1)}$. We will show how to reduce the space complexity to $n^{2+1/c+o(1)}$ in \cref{sec:reduce-space}.
	
	\subsection{Data Structure}
	Let $\epsilon_1=\epsilon/(2+\epsilon),\epsilon_2=\epsilon_1/(2k-1)$, so $\epsilon_1<1$. (Recall $k=\ln n$.) We first define a decomposition of a path $P$ into $O(\log|P|/\epsilon_2)$ segments. This definition has the same spirit as \cite[Definition 2.1]{CCFK17}, but it partitions the \emph{vertices} (excluding $u,v$), rather than \emph{edges}, into segments.
	\begin{definition}[($\epsilon_2$-)segments]
		Consider a path $P=(u=v_0,v_1,v_2,\dots,v_{\ell}=v)$. For every $1\le i,j<\ell$, we say $v_i$ and $v_j$ are in the same \emph{segment} if  one of the two conditions hold:
		\begin{itemize}
		    \item $|P[u,v_i]|,|P[u,v_j]|\leq |P|/2$ and $\lfloor\log_{1+\epsilon_2}|P[u,v_i]|\rfloor=\lfloor\log_{1+\epsilon_2}|P[u,v_j]|\rfloor$.
		    \item $|P[v_i,v]|,|P[v_j,v]|< |P|/2$ and $\lfloor\log_{1+\epsilon_2}|P[v_i,v]|\rfloor=\lfloor\log_{1+\epsilon_2}|P[v_j,v]|\rfloor$.
		\end{itemize}\label{def:segments}
	\end{definition}
	\begin{figure}
		\centering
		\includegraphics[width=0.8\linewidth]{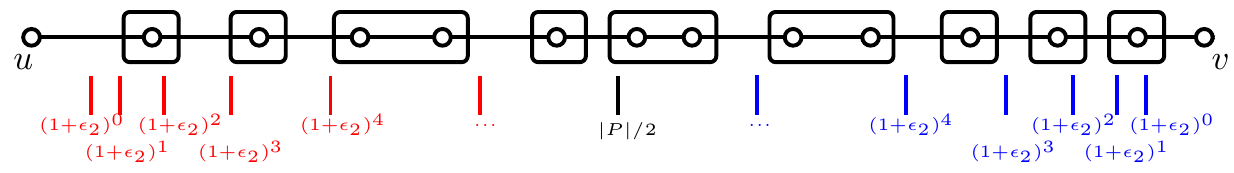}
		\caption{An illustration of path decomposition, where each rounded rectangle denotes a segment.}
		\label{fig:segments}
	\end{figure}
	
	It is easy to verify that being in the same segment is indeed an equivalence relation, and each equivalence class is indeed a contiguous segment of the path. (See \cref{fig:segments}.) There are $O(\log|P|/\epsilon_2)$ different segments. For a vertex $v_i$ on $P$, define $\seg(v_i,P)$ as the segment it belongs to. More precisely, $\seg(v_i,P)=P[v_l,v_r]$ where $v_l$ is the leftmost (closest-to-$u$) vertex in the segment, and $v_r$ is the rightmost (closest-to-$v$) vertex in the segment. Define $\seg(P)$ as the set of segments on $P$. Note that $u$ and $v$ do not belong to any segment.
	
	\begin{lemma}
		Let $P$ be a path from $u$ to $v$, $x\in P\setminus\{u,v\}$, then $|\seg(x,P)|\le\epsilon_2\cdot\min\{|P[u,x]|,|P[x,v]|\}$.\label{lemma:epsilon2}
	\end{lemma}
	\begin{proof}
		If $|P[u,x]|\leq |P|/2$, $|P[u,x]|\leq |P[x,v]|$ and $\seg(x,P)$ is defined in the first way. Let $i=\lfloor\log_{1+\epsilon_2}|P[u,x]|\rfloor$. For any $y\in\seg(x,P)$, $(1+\epsilon_2)^i\le|P[u,y]|<(1+\epsilon_2)^{i+1}$. Thus $|\seg(x,P)|\le(1+\epsilon_2)^{i+1}-(1+\epsilon_2)^i=\epsilon_2(1+\epsilon_2)^i\le\epsilon_2|P[u,x]|$. The case that $|P[x,v]|<|P|/2$ is symmetric.
	\end{proof}
	
	Recall that we build a data structure for every node $U_p$ in the hierarchy tree. Let the path from the root $V$ to $U_p$ be $U_1(=V),U_2\dots,U_p$ and $U_{p+1}=\varnothing$. The data structure for $U_p$ consists of decision trees $FT(u,v)$ for all pairs of vertices $u,v\in V$, which are constructed as follows:
	\begin{itemize}
		\item Each node\footnote{Recall that we use ``vertex'' for nodes in the input graph, and ``node'' for nodes in the decision trees and the hierarchy tree.} $\alpha\in FT(u,v)$ is associated with a set $\avoid(\alpha)\subseteq V$ of vertices that we avoid.
		\item Denote the root of $FT(u,v)$ as $\root(u,v)$, and let $\avoid(\root(u,v))=\varnothing$.
		\item For each node $\alpha\in FT(u,v)$, we store the path $P_\alpha=\pi_{G-\avoid(\alpha)}(u,v)$, \ie the shortest $u$-$v$ path in $G$ not passing through $\avoid(\alpha)$. If $u,v$ are not connected in $G-\avoid(\alpha)$, we assume that $P_\alpha$ is a path with length $+\infty$.
		\item For each node $\alpha$ of depth $<d$ (the root has depth $0$), each segment $X\in\seg(P_\alpha)$, and each $1\le i\le p$, we create a child $ch(\alpha,X,i)$ of $\alpha$, in which
			\begin{equation}
				\avoid(ch(\alpha,X,i))=\avoid(\alpha)\cup (X\cap U_i). \label{eq:avoid-children}
			\end{equation}
		That is, the path stored in a child of $\alpha$ needs to avoid $\avoid(\alpha)$ and the vertices of $U_i$ in a segment $X\in\seg(P_\alpha)$.
	\end{itemize}

	The decision tree has depth $d$, and each non-leaf node has $O(p\cdot \log|P|/\epsilon_2)=O(h\epsilon^{-1}\log n\log(nW))$ children. (Recall $h=O(\log n/\log d)$ and $\epsilon_2=\epsilon/((2+\epsilon)(2k-1))=\epsilon/\Theta(\log n)$.) We store in each node $\alpha$ the path $P_\alpha$ as well as a table of $\seg(v,P_\alpha)$ for each $v\in P_\alpha$, so that we can quickly locate any vertex in $P_{\alpha}$. Therefore one decision tree occupies $n\cdot O(h\epsilon^{-1}\log n\log(nW))^d$ space. As there are $O(n^{1/c})$ nodes in the hierarchy tree, and for each node we need to store $O(n^2)$ decision trees, the total space complexity is $n^{3+1/c}\cdot O(h\epsilon^{-1}\log n\log(nW))^d$.
	
	\subsection{Query Algorithm}\label{sec:query}
	Given $u,v$ and a set of failed vertices $D$, by \cref{lemma:good-path} we first find a path $U_1(=V),U_2\dots,U_p$ in the hierarchy tree, and set $U_{p+1}=\varnothing$. Let $\caT=\bigcup_{i=1}^p\caT_{U_{i+1}}(U_i)$, then by \cref{lemma:good-path}, the pseudo-degrees of all $f\in D$ in every tree in $\caT$ is at most $s$.
	
	As in~\cite{CCFK17}, the query algorithm builds an auxiliary graph $H$, but the definition of $H$ is different from~\cite{CCFK17}. The query algorithm builds $H$ as \cref{def:VH}, and outputs $|\pi_H(u,v)|$ as an $(1+\epsilon)$-approximation of $|\pi_{G-D}(u,v)|$.\footnote{The vertex set of $H$ corresponds to the set of ``intermediate vertices'' (also called $H$) in \cref{sec:brief-overview}.}
	
	\begin{definition} (Graph $H$)
		\begin{itemize}
			\item For a failure $f\in D$ and a tree $T\in\caT$, if $f\in V(T)$, then we define the neighbors of $f$ in $T$ as
			\[N_T(f)=\{\parent_T(f)\}\cup(\children_T(f)\cap\Trunk(T)).\]
			In other words, $N_T(f)$ consists of the parent of $f$ in $T$, and the set of children of $f$ in $T$ which are trunk vertices. (Note that if $f\in V(T)\setminus\Trunk(T)$, then it is possible that $\parent_T(f)\not\in\Trunk(T)$.)
			\item Define \[N(f)=\bigcup_{f\in T\in\caT}N_T(f).\]
			That is, $N(f)$ is the union of $N_T(f)$'s over all trees $T\in\caT$ such that $f\in T$.
			\item The vertex set of the auxiliary graph $H$ is
			\[V(H)=\left(\{u,v\}\cup\bigcup_{f\in D}N(f)\right)\setminus D.\]
			For each $x,y\in V(H)$, the weight of the edge $(x,y)$ in $H$ is equal to $\DecTree(x,y,D)$, as defined in \cref{alg:DecTree}.
		\end{itemize}
		\label{def:VH}
	\end{definition}

	Note that in $V(H)$, vertices except $u$ and $v$ are defined independently from $u$ and $v$. By \cref{lemma:good-path}, for every $f\in D$ and $T\in\caT$ containing $f$, we have $|N_T(f)|\le s+1$. Therefore $|V(H)|\le dp\cdot 2e\ln^2 n\cdot (s+1)+2=O(d^{c+2}h\log^4 n)$. 

	\begin{algorithm}[H]
	\caption{Algorithm $\DecTree$\label{alg:DecTree}}
		\begin{algorithmic}[1]
			\Function{$\DecTree$}{$u,v,D$}
			\State $\alpha\gets \root(u,v)$
			\While {$D\cap P_\alpha\ne\varnothing$}
				\State {$f\gets$ the vertex in $D\cap P_\alpha$ with the highest level, breaking ties arbitrarily}\label{line:choose-f}
				\State {$\alpha\gets ch(\alpha,\seg(f,P_\alpha),l(f))$} \label{line:del-DecTree} \Comment{Recall $l(f)$ is the level of $f$, \ie the largest $l$ s.t.~$f\in U_l$.}
			\EndWhile
			\State\Return $|P_\alpha|$
			\EndFunction
		\end{algorithmic}
	\end{algorithm}
	
	Consider \cref{alg:DecTree}. After each iteration, the set $D\cap\avoid(\alpha)$ will contain at least one new vertex (namely $f$). When $|D\cap\avoid(\alpha)|=d$, we will have that $D\cap P_\alpha=\varnothing$, and the algorithm terminates. Thus the algorithm executes at most $d$ iterations. It is easy to see that each iteration only requires $O(d)$ time, thus the time complexity of \cref{alg:DecTree} is $O(d^2)$.
	
	It takes $O(|V(H)|^2d^2)$ time to build the graph $H$, and $O(|V(H)|^2)$ time to compute $|\pi_H(u,v)|$. Therefore, the query algorithm runs in $O(|V(H)|^2d^2)=O(d^{2c+6}h^2\log^8n)$ time.
	
	Since finally $D\cap P_{\alpha}=\emptyset$, we have the following observation:
	\begin{observation}\label{lemma:longer}
		For all $u,v\in V(H)$, $\DecTree(u,v,D)\geq |\pi_{G-D}(u,v)|$.
	\end{observation}

	For every $\alpha$ and $f$ considered in \cref{alg:DecTree}, let $\alphanext=ch(\alpha,\seg(f,P_\alpha),l(f))$ be the next decision tree node the algorithm considers. If the optimal path $\pi_{G-D}(u,v)$ never intersects the set $\avoid(\alphanext)\setminus\avoid(\alpha)$ in any iteration, then $\DecTree(u,v,D)$ would indeed return the optimal path $\pi_{G-D}(u,v)$. However, if $\pi_{G-D}(u,v)$ goes through some non-failure vertex $x\in\avoid(\alphanext)\setminus\avoid(\alpha)$, then $x$ is close to $f$, and we will show that there is some $w\in V(H)$ close to $x$, so the optimal path can be approximated by $\pi_{G-D}(u,w)\circ \pi_{G-D}(w,v)$. This is illustrated in the following important lemma. Notice that $\avoid(\alphanext)\setminus\avoid(\alpha)=\seg(f,P_\alpha)\cap U_{l(f)}$, so in this case, the level of $x$ is no less than the level of $f$, \ie $l(x)\ge l(f)$.
	
	\begin{lemma}\label{lemma:ball_VH}
	    Given a failed vertex $f$ and a non-failure vertex $x$ such that $l(x)\geq l(f)$, if there is a path $P$ in $G$ between $x$ and $f$ which contains no other failed vertices, then there is a vertex $w\in V(H)$ such that $|\pi_{G-D}(x,w)|\leq (2k-1)|P|$.
	\end{lemma}

	\begin{figure}[H]
		\centering
		\includegraphics[scale=1]{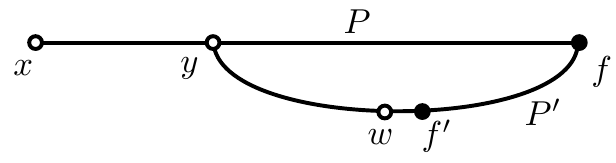}
	\end{figure}

	\begin{proof}
	    Let $y$ be the vertex with the highest level on $P[x,f)$, and suppose $j=l(y)$. Then $l(y)=j\geq l(x)\geq l(f)$. Since there are no vertices on $P$ with level $>j$, $P$ is in $G-U_{j+1}$. Let $T_j(y,f)$ be the tree in the tree cover $\caT_{U_{j+1}}(U_j)$, such that the distance between $y$ and $f$ in $T_j(y,f)$ is at most $(2k-1)|\pi_{G-U_{j+1}}(y,f)|$. Let $P'=\caP_j(y,f)$ be the path between $y$ and $f$ in $T_j(y,f)$. Let $f'$ be the first failed vertex on $P'[y,f]$, and consider the predecessor $w$ of $f'$ on the path $P'[y,f]$. (That is, $P'[y,w]$ is intact from failures.) Since $P'$ is a path on the tree $T_j(y,f)$, $w$ is either the parent of $f'$ or a child of $f'$ in this tree. \begin{itemize}
	    	\item If $w$ is the parent of $f'$, then $w\in V(H)$ by the definition of $V(H)$.
	    	\item If $w$ is a child of $f'$, then $y$ is a descendant of $w$ and $f'$. Since $l(y)=j$, $y$ is a trunk vertex in $T_j(y,f)$. It follows that $w$, as an ancestor of $y$, is also a trunk vertex in $T_j(y,f)$, therefore $w\in V(H)$.
	    \end{itemize}
		Therefore, in either case, we have $w\in V(H)$. Since
	    \begin{align}
	    |\pi_{G-D}(x,w)|\le&~ |P[x,y]|+|P'[y,w]| \nonumber\\
		\le&~|P[x,y]|+ |P'| \nonumber\\
		\le&~|P[x,y]|+ (2k-1)|\pi_{G-U_{j+1}}(y,f)|\nonumber\\
		\le&~|P[x,y]|+ (2k-1)|P[y,f]| \nonumber\\
		\le&~(2k-1)|P|, \nonumber
		\end{align}
		the lemma is true.
	\end{proof}

	\subsection{Proof of Correctness}
	In this section, we show that $|\pi_H(u,v)|\le (1+\epsilon)|\pi_{G-D}(u,v)|$, proving the correctness of the query algorithm.
	
	From the algorithm $\DecTree(u,v,D)$, the path we get is the shortest path between $u$ and $v$ in the graph $G-\avoid(\afinal)$, where $\afinal$ is the last visited decision tree node of the algorithm. As we discussed before, if the real shortest path $\pi_{G-D}(u,v)$ does not go through any vertex in $\avoid(\afinal)$, then $\DecTree(u,v,D)$ will return the correct answer. Otherwise, as $\avoid(\afinal)$ is the union of $\leq d$ sets of the form $\seg(f,P_\alpha)\cap U_i$, $\pi_{G-D}(u,v)$ must go through some vertex $x$ in a set $\seg(f,P_\alpha)\cap U_i$. We can show that such $x$ will be ``close'' to a vertex $w$ in $V(H)$ (by \cref{lemma:ball_VH}), so we can use the vertices in $V(H)$ as intermediate vertices to obtain an approximate shortest path.
	
	\begin{lemma}\label{lemma:VH}
	    In the query algorithm $\DecTree(u,v,D)$, let $\alpha$ be a decision tree node it encounters, $f\in D$ be the failed vertex which is selected in \cref{line:choose-f} of \cref{alg:DecTree} and $i=l(f)$. (That is, $f$ is the vertex in $D\cap P_\alpha$ with the highest level.) For any non-failure vertex $x$ in $\seg(f,P_\alpha)\cap U_i$, there is a vertex $w\in V(H)$ such that $|\pi_{G-D}(x,w)|\leq \epsilon_1 \min\{|P_\alpha[u,x]|,|P_\alpha[x,v]|\}$.
	\end{lemma}
	
	\begin{figure}[H]
		\centering
		\includegraphics{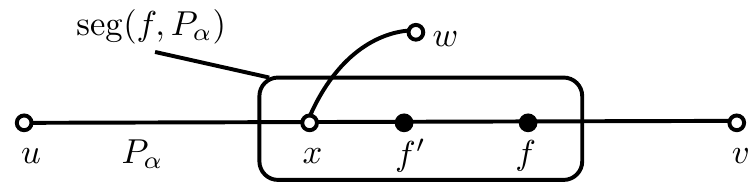}
	\end{figure}

	\begin{proof}
	    Let $f'\in D$ be the failed vertex closest to $x$ on the segment $\seg(f,P_\alpha)$, then there are no failed vertices in $P_\alpha[x,f')$ or $P_\alpha(f',x]$. (W.l.o.g.~we assume it is $P_\alpha[x,f')$.) We have $l(f)\ge l(f')$ by \cref{alg:DecTree}. As $x\in U_i=U_{l(f)}$, we have $l(x)\ge l(f)$, hence $l(x)\ge l(f')$. By \cref{lemma:ball_VH}, there is a vertex $w\in V(H)$ such that $|\pi_{G-D}(x,w)|\leq (2k-1)|P_\alpha[x,f']|$. Since $|P_\alpha[x,f']|\leq \epsilon_2\min\{|P_\alpha[u,x]|,|P_\alpha[x,v]|\}$ and $\epsilon_2=\epsilon_1/(2k-1)$, the lemma holds.
	\end{proof}

	We show that for $u,v\in V(H)$, if the optimal path $\pi_{G-D}(u,v)$ is \emph{not} found by $\DecTree(u,v,D)$, then we can indeed find some $w\in V(H)$ such that $\pi_{G-D}(u,w)\circ\pi_{G-D}(w,v)$ is a good approximation of $\pi_{G-D}(u,v)$. Moreover, one of $\pi_{G-D}(u,w)$ or $\pi_{G-D}(w,v)$ can be dealt with by \cref{alg:DecTree}, therefore we only need to ``recurse'' on the other one.
	
	\begin{lemma}\label{lemma:pair}
	    Let $u,v\in V(H)$ and $P=\pi_{G-D}(u,v)$. If $\DecTree(u,v,D)>|P|$, there exist $x\in P\setminus\{u,v\},y\in\{u,v\},w\in V(H)$ such that
		\begin{enumerate}[(a)]
			\item $|\pi_{G-D}(x,y)|\le \frac{1}{2}|\pi_{G-D}(u,v)|$,\label{item:pair-1}
			\item $|\pi_{G-D}(x,w)|\leq \epsilon_1|\pi_{G-D}(x,y)|$, and\label{item:pair-2}
			\item $\DecTree(y,w,D)\leq |\pi_{G-D}(y,x)|+|\pi_{G-D}(x,w)|$, which is smaller than $|\pi_{G-D}(u,v)|$.\label{item:pair-3}
		\end{enumerate}
	\end{lemma}
	
	\begin{figure}[H]
		\centering
		\includegraphics{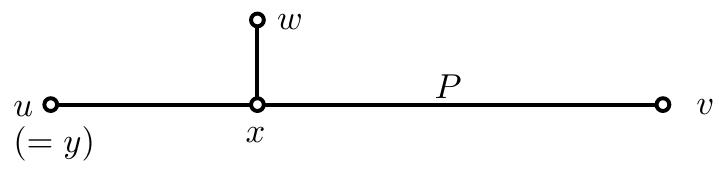}
	\end{figure}
	
	\begin{proof}
		First we prove that there exists a triple $(x,y,w)$ that satisfies (\ref{item:pair-1}) and (\ref{item:pair-2}).
		
		Let $\alpha$ be the last decision tree node visited by $\DecTree(u,v,D)$ such that $\avoid(\alpha)\cap P=\varnothing$. Since $\DecTree(u,v,D)>|P|$, the procedure $\DecTree(u,v,D)$ did not terminate at $\alpha$, \ie it visited a child $\alphanext$ of $\alpha$ such that $P\cap\avoid(\alphanext)\ne\varnothing$. Recall that $\avoid(\alphanext)\setminus\avoid(\alpha)=\seg(f,P_\alpha)\cap U_{l(f)}$ where $f$ is the failure selected by \cref{line:choose-f} of \cref{alg:DecTree}. Therefore $P$ reaches some vertex $x\in \seg(f,P_{\alpha})\cap U_{l(f)}$. Since $P_{\alpha}$ is the shortest $u$-$v$ path in $G-\avoid(\alpha)$ and $P$ is \emph{some} $u$-$v$ path in $G-\avoid(\alpha)$, we know that $|P_{\alpha}[u,x]|\le |P[u,x]|$ and $|P_{\alpha}[x,v]|\le |P[x,v]|$.
		
		By \cref{lemma:VH}, there is a vertex $w\in V(H)$ such that \[
		|\pi_{G-D}(x,w)|\le \epsilon_1\min\{|P_{\alpha}[u,x]|,|P_{\alpha}[x,v]|\}\le\epsilon_1\min\{|P[u,x]|,|P[x,v]|\}.
		\]
		Let $y$ be the endpoint in $\{u,v\}$ that is closer to $x$, then $(x,y,w)$ satisfies (\ref{item:pair-1}) and (\ref{item:pair-2}).
		
		Among all triples $x\in P\setminus\{u,v\},y\in\{u,v\},w\in V(H)$ satisfying (\ref{item:pair-2}), we pick a triple minimizing $|\pi_{G-D}(x,y)|$, and in case of a tie choose a triple minimizing $|\pi_{G-D}(x,w)|$. It is easy to see that (\ref{item:pair-1}) is also satisfied. In the following we prove that (\ref{item:pair-3}) is satisfied.
	    
		We compare the path $P'=\pi_{G-D}(y,x)\circ\pi_{G-D}(x,w)$ between $y$ and $w$, with the path returned by $\DecTree(y,w,D)$. For the sake of contradiction, suppose $|P'|<\DecTree(y,w,D)$. Let $\alpha'$ be the last decision tree node visited in $\DecTree(y,w,D)$ such that $\avoid(\alpha')\cap P'=\varnothing$. We can also see that $P'$ reaches some vertex $x'\in\seg(f',P_{\alpha'})\cap U_{l(f')}$, where $f'$ is the failure selected by \cref{line:choose-f} of \cref{alg:DecTree}. We use \cref{lemma:VH} again and conclude that there is a vertex $w'\in V(H)$ such that
		\[|\pi_{G-D}(x',w')|\le\epsilon_1\min\{|P_{\alpha'}[y,x']|,|P_{\alpha'}[x',w]|\}\le\epsilon_1\min\{|P'[y,x']|,|P'[x',w]|\}.\]

		Since $x'\in P'=P'[y,x]\circ P'[x,w]$, there are two cases. (See \cref{figure:pair}.)
		
		\begin{figure}[H]
			\centering
			\includegraphics[scale=0.7]{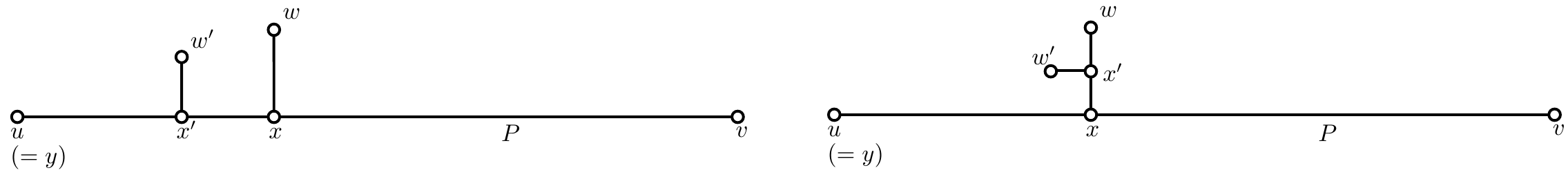}
			\caption{Two cases in the proof of \cref{lemma:pair}.}\label{figure:pair}
		\end{figure}
		
		\begin{itemize}
			\item If $x'\in P'[y,x)$, then $|\pi_{G-D}(x',w')|\le\epsilon_1|\pi_{G-D}(x',y)|$, \ie the triple $(x',y,w')$ also satisfies (\ref{item:pair-2}). Since $|\pi_{G-D}(x',y)|<|\pi_{G-D}(x,y)|$, this contradicts our choice of $(x,y,w)$.
			\item If $x'\in P'[x,w]$, then $|\pi_{G-D}(x,w')|\leq |P'[x,x']|+|\pi_{G-D}(x',w')|\leq |P'[x,x']|+\epsilon_1|P'[x',w]|$. As $\epsilon_1<1$, we have $|\pi_{G-D}(x,w')|<|\pi_{G-D}(x,w)|$, and $(x,y,w')$ also satisfies (\ref{item:pair-2}), contradicting our choice of $(x,y,w)$.
		\end{itemize}
	
		Hence it must be true that $|P'|\ge \DecTree(y,w,D)$.
	\end{proof}

	By these lemmas, we can now prove our desired approximation ratio.

	\begin{theorem}\label{thm:main-approx}
		For every pair $u,v\in V(H)$, the query algorithm in \cref{sec:query} returns an $(1+\epsilon)$-approximation of $|\pi_{G-D}(u,v)|$.
	\end{theorem}
	\begin{proof}
		It is easy to see that $|\pi_H(u,v)|\ge |\pi_{G-D}(u,v)|$ for every $u,v\in V(H)$. We prove $|\pi_H(u,v)|\le (1+\epsilon)|\pi_{G-D}(u,v)|$ below.
		
		We sort all pairs of vertices $u,v\in V(H)$ ($u\neq v$) by increasing order of $|\pi_{G-D}(u,v)|$, and prove by induction that $|\pi_H(u,v)|\le(1+\epsilon)|\pi_{G-D}(u,v)|$ on this order. For the $u,v$ having the smallest $|\pi_{G-D}(u,v)|$, if $\DecTree(u,v,D)>|\pi_{G-D}(u,v)|$, from \cref{lemma:pair}, there exist $y,w\in V(H)$ so that $|\pi_{G-D}(y,w)|<|\pi_{G-D}(u,v)|$, which is a contradiction. Therefore $|\pi_H(u,v)|=\DecTree(u,v,D)=|\pi_{G-D}(u,v)|$. 
	    
	    Fix some $u,v\in V(H)$, assume that for all pairs $u',v'\in V(H)$ such that $|\pi_{G-D}(u',v')|<|\pi_{G-D}(u,v)|$, it is true that $|\pi_H(u',v')|\le(1+\epsilon)|\pi_{G-D}(u',v')|$. If $\DecTree(u,v,D)=|\pi_{G-D}(u,v)|$ then $|\pi_H(u,v)|\le(1+\epsilon)|\pi_{G-D}(u,v)|$ follows trivially. Otherwise we use \cref{lemma:pair} to obtain a triple $(x,y,w)$, where $x\in \pi_{G-D}(u,v)\setminus\{u,v\}$, $y\in\{u,v\}$, and $w\in V(H)$. We assume w.l.o.g.~$y=u$, then $|\pi_{G-D}(w,x)|\le\epsilon_1|\pi_{G-D}(u,x)|$. Since $\epsilon_1<1$, $|\pi_{G-D}(w,v)|\leq |\pi_{G-D}(w,x)|+|\pi_{G-D}(x,v)|\leq \epsilon_1|\pi_{G-D}(u,x)|+|\pi_{G-D}(x,v)|<|\pi_{G-D}(u,v)|$, thus $|\pi_{H}(w,v)|\le (1+\epsilon)|\pi_{G-D}(w,v)|$ by induction hypothesis. We have:
	    \begin{align*}
	        |\pi_H(u,v)| \le &~ \DecTree(u,w,D)+|\pi_H(w,v)| \\
	        \le &~ |\pi_{G-D}(u,x)|+|\pi_{G-D}(x,w)|+(1+\epsilon)|\pi_{G-D}(w,v)| \\
	        \le &~ |\pi_{G-D}(u,x)|+\epsilon_1|\pi_{G-D}(u,x)|+(1+\epsilon)(\epsilon_1|\pi_{G-D}(u,x)|+|\pi_{G-D}(x,v)|) \\
	        \le &~ (1+\epsilon_1+(1+\epsilon)\epsilon_1)|\pi_{G-D}(u,x)|+(1+\epsilon)|\pi_{G-D}(x,v)|\\
	        \le &~ (1+\epsilon)|\pi_{G-D}(u,v)|.\qedhere
	    \end{align*}
	\end{proof}
	
	We conclude that there is a VSDO with
	\[\SpaceQueryStretch{n^{3+1/c}\cdot O(\epsilon^{-1}\log^2 n\log(nW)/\log d)^d}{O(d^{2c+6}\log^{10}n/\log^2 d)}{1+\epsilon.}\]
	
	In \cref{sec:reduce-space}, we will improve the space complexity to $n^{2+1/c+o(1)}$ when $d=o\mleft(\frac{\log n}{\log\log n}\mright)$ and $W=\poly(n)$, while increasing the query time slightly.

\section{An $n^{2+1/c+o(1)}$-Space $(1+\epsilon)$-Stretch Oracle}\label{sec:reduce-space}
	In this section, we discuss the modifications needed to reduce the space complexity to $n^{2+1/c+o(1)}$. Here we set $\epsilon_3=\frac{\epsilon}{2|V(H)|}$, and $\epsilon_4=\epsilon_3/(4k-2)$, where $V(H)$ is the same as in \cref{sec:orac2} (and \cref{lemma:ball_VH} still holds), while $E(H)$ are recomputed in this section. We assume that $\epsilon$ is small enough, in particular that $\epsilon<1$ and $\epsilon_4<\sqrt{2}-1$.
	
	\subsection{A Structural Theorem}
	Similar to \cite{CCFK17}, the main idea is, instead of storing the paths $P_\alpha$ as-is in every node $\alpha$ of $FT(u,v)$, we store an implicit representation of these paths. If the representation has size $\poly(\log (nW),\epsilon^{-1})$ instead of $\Omega(n)$, then our data structure has space complexity $n^{2+1/c+o(1)}$.
	
	In \cite{CCFK17}, the authors defined \emph{$k$-decomposable paths}, which are paths that can be represented as the concatenation of at most $k+1$ shortest paths in $G$, interleaved with at most $k$ edges. They relied on the fact (Theorem 2 of \cite{ABKCM02}) that any $k$-edge-failure shortest path is a $k$-decomposable path in $G$, therefore has a succinct representation. Unfortunately, the analogue of this statement in~\cite{CCFK17} in case of vertex failures does not hold. Even if we only remove \emph{one} vertex (\ie $|D|=1$), a shortest path in $G-D$ might not be a $k$-decomposable path for $k=o(n)$.\footnote{Consider an unweighted graph $G=(V,E_1\cup E_2)$ where $V=\{v_i:0\le i\le n\}$, $E_1=\{(v_0,v_i):1\le i\le n\}$ and $E_2=\{(v_i,v_{i+1}):1\le i<n\}$. Then $\pi_{G-\{v_0\}}(v_1,v_n)$ is not a $0.1n$-decomposable path.}
	
	In this section, we prove a structural theorem similar to the above fact used in \cite{CCFK17}. Before we proceed, we need some definitions.
	
	From \cref{lemma:pair} we can see that for any $u,v\in V(H)$, if the path $\pi_{G-D}(u,v)$ is \emph{$\epsilon_1$-far away from} $V(H)$ in the following sense, then $\DecTree(u,v,D)$ indeed finds the distance between $u$ and $v$ in $G-D$:
	\begin{definition}
		We say that a path $P$ from $u$ to $v$ is \emph{$\epsilon$-far away from $V(H)$} if there are no vertices $x\in P\setminus \{u,v\}, w\in V(H)$ such that $|\pi_{G-D}(x,w)|\leq \epsilon\cdot\min\{|P[u,x]|,|P[x,v]|\}$. (See \cref{fig:far-away}.)
	\end{definition}
	
	\begin{figure}
		\centering
		\includegraphics{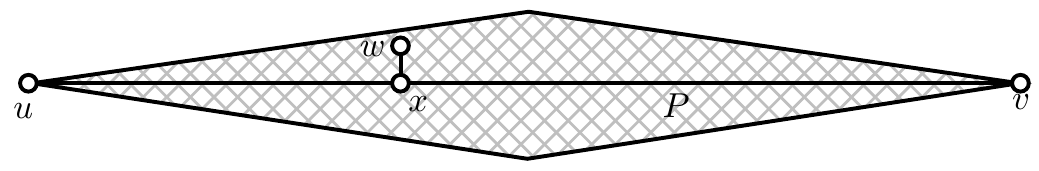}
		\caption{If $P$ is far away from $V(H)$, it means that a certain ``diamond''-shaped area does not contain vertices $w\in V(H)$.}
		\label{fig:far-away}
	\end{figure}
	
	Instead of considering all $d$-failure shortest paths, we only study the ones which are $\epsilon_3$-far away from $V(H)$. We will use the concept of $k$-expath as in~\cite{CCFK17} and re-define it as \emph{$\epsilon_4$-segment expath}. Also, instead of considering the concatenation of at most $k+1$ shortest paths in the original graph $G$, every segment here is a shortest path in some $G_i$. (Recall that $G_i$ is the induced subgraph of $G$ on all vertices of level $\le i$.)
	
	\begin{definition}\label{def:segment_expath}
		A path $P$ in $G$ is an \emph{$\epsilon$-segment expath} if the following holds. If we partition $P$ into $\epsilon$-segments as in \cref{def:segments}, then for every segment $P[x,y]$, there is some $1\le i\le p$ such that $P[x,y]$ is a shortest path in $G_i$.
	\end{definition}
	
	The following structural theorem for shortest paths $\epsilon$-far away from $V(H)$ will be crucial to us. Interestingly, it is a consequence of \cref{lemma:ball_VH}.
	
	\begin{theorem}\label{thm:decomposable}
		For $u,v\in V(H)$, if $\pi_{G-D}(u,v)$ is $\epsilon_3$-far away from $V(H)$, then it is an $\epsilon_4$-segment expath.
	\end{theorem}
	
	\begin{figure}[H]
		\centering
		\includegraphics{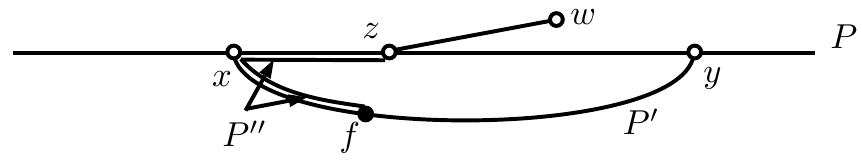}
	\end{figure}
	
	\begin{proof}
		Let $P=\pi_{G-D}(u,v)$ and $P[x,y]$ be an $\epsilon_4$-segment of $P$. W.l.o.g.~assume that $x$ and $y$ are in the first half of $P$, and $x$ is closer to $u$ than $y$. By \cref{lemma:epsilon2}, $|P[x,y]|\le\epsilon_4 |P[u,x]|$. (Recall that $P[x,y]$ is an $\epsilon_4$-segment.) Consider the vertex $z$ with the highest level on $P[x,y]$, and let its level be $i=l(z)$. Then $P[x,y]$ is a path in $G_i$. If it is not the shortest path $\pi_{G_i}(x,y)$, then $\pi_{G_i}(x,y)$ must go through some failed vertex in $D$. (Since otherwise we can find a path in $G-D$ shorter than $\pi_{G-D}(u,v)$.) Let $P'=\pi_{G_i}(x,y)$ and $f$ be the failed vertex on $P'$ closest to $x$.
		
		Since $f$ is in the graph $G_i$, we have $l(f)\le i=l(z)$. There is a path $P''=P[z,x]\circ P'[x,f]$ connecting $z$ and $f$ that does not go through other failed vertices. By \cref{lemma:ball_VH}, there is a vertex $w\in V(H)$ such that $|\pi_{G-D}(z,w)|\le (2k-1)|P''|$. We have
		\begin{align*}
			|\pi_{G-D}(z,w)| \le&~ (2k-1)(|P[z,x]|+|P'[x,f]|)\\
			\le &~ 2(2k-1)|P[x,y]|\\
			\le &~ 2(2k-1)\epsilon_4|P[u,x]|\\
			\le &~\epsilon_3|P[u,z]|,
		\end{align*}
		which contradicts that $\pi_{G-D}(u,v)$ is $\epsilon_3$-far away from $V(H)$. Therefore, $P[x,y]$ is a shortest path in $G_i$.
	\end{proof}
	
	\subsection{New Data Structure}\label{sec:new_query}
	
	We generalize the concept of $\epsilon_4$-segment expath to $\epsilon_4$-expath by adding more flexibility.
	
	\begin{definition}\label{def:expath}
		Let $B=\lceil\log_{1+\epsilon_4}(nW)\rceil$. An \emph{$\epsilon_4$-expath} $P$ from $u$ to $v$ in $G$ is a path which is a concatenation of subpaths $P_0,\dots, P_{2B+1}$ interleaved with at most $2B+3$ edges\footnote{That is, the concatenation of $e_0,P_0,e_1,P_1,\dots,P_{2B+1},e_{2B+2}$ where each $e_i$ is either empty or an edge.}, such that the following hold. \begin{itemize}
			\item For every $P_k=P[u_k,v_k]$ ($0\leq k\le 2B+1$), $P_k$ is either empty, or a shortest path in $G_i$ for some level $1\le i\le p$.
			\item If $k<B+1$, then $|P[u,v_k]|\leq (1+\epsilon_4)^{k}$; if $k\ge B+1$, then $|P[u_k,v]|\leq (1+\epsilon_4)^{2B+1-k}$.
		\end{itemize}
	\end{definition}
	
	\begin{lemma}
		An $\epsilon_4$-segment expath $P$ from $u$ to $v$ is an $\epsilon_4$-expath.\label{lemma:expath}
	\end{lemma}
	\begin{proof}
		Let $j=\lfloor\log_{1+\epsilon_4}(|P|/2)\rfloor+1$, since $1+\epsilon_4<2$, we have $j<B+1$. Let $P_1,\dots,P_j$ be the $\epsilon_4$-segments (possibly empty) in the first half of $P$ such that for every $1\le k\le j$ and $x\in P_k=P[u_k,v_k]$, $\lfloor\log_{1+\epsilon_4}|P[u,x]|\rfloor=k-1$. Then $|P[u,v_k]|\le (1+\epsilon_4)^k$, which satisfies the definition of $\epsilon_4$-expath. The second half of $P$ is symmetric.
	\end{proof}
	
	Recall that our data structure in \cref{sec:orac2} consists of $O(n^2)$ decision trees, one for each pair $u,v\in V$. Each decision tree node $\alpha$ stores a path $P_\alpha$, a subset $\avoid(\alpha)$ of $V$, and the links to its children. The query algorithm builds an auxiliary graph $H$ on the vertex set $V(H)$ defined in \cref{def:VH}, and uses \cref{alg:DecTree} to determine the edge weights in $H$. At last we output $|\pi_H(u,v)|$ as the approximation of $|\pi_{G-D}(u,v)|$. Our improved data structure also fits into this high-level description, but there are some small changes:
	\begin{itemize}
		\item For every $1\le i\le p$, we also store the shortest path distance matrix of $G_i$.
		\item We use $\epsilon_4$ in the definition of segments.
		\item In every node $\alpha\in FT(u,v)$, we store the shortest $\epsilon_4$-expath (instead of the general shortest path) from $u$ to $v$ in $G-\avoid(\alpha)$, still denoted as $P_{\alpha}$. To save space, for every subpath $P_k=[u_k,v_k]$ which is a shortest path in some $G_i$, we only need to store a triple $(u_k,v_k,i)$.
		\item To check whether $f$ is in a path $P_{\alpha}$, for every subpath $P_k=[u_k,v_k]$ which is a shortest path in some $G_i$, we check whether $|\pi_{G_i}(u_k,f)|+|\pi_{G_i}(f,v_k)|=|\pi_{G_i}(u_k,v_k)|$. By the uniqueness assumption of shortest paths (see~\cite{DTCR08}), this method can locate a vertex in $P_{\alpha}$.
	\end{itemize}
	
	We now prove the correctness of this data structure, \ie $|\pi_H(u,v)|$ is always an $(1+\epsilon)$-approximation of $|\pi_{G-D}(u,v)|$.
	
	First, it is easy to check that \cref{lemma:VH} holds for parameter $(2k-1)\epsilon_4=\epsilon_3/2$, as follows. 
	
	\begin{reminder}{\cref{lemma:VH}}
		In the query algorithm $\DecTree(u,v,D)$, let $\alpha$ be a decision tree node it encounters, $f\in D$ be the failed vertex which is selected in \cref{line:choose-f} of \cref{alg:DecTree} and $i=l(f)$. (That is, $f$ is the vertex in $D\cap P_\alpha$ with the highest level.) For any non-failure vertex $x$ in $\seg(f,P_\alpha)\cap U_i$, there is a vertex $w\in V(H)$ such that $|\pi_{G-D}(x,w)|\leq (\epsilon_3/2) \min\{|P_\alpha[u,x]|,|P_\alpha[x,v]|\}$.
	\end{reminder}
	
	Recall that \cref{lemma:pair} shows that, in the data structure in \cref{sec:orac2}, any shortest path $\epsilon_1$-far away from $V(H)$ can be found by $\DecTree$. We show that this is also true in the new data structure, where ``$\epsilon_1$-far away'' is changed to ``$\epsilon_3$-far away''.
	
	\begin{lemma}\label{lemma:VH2}
		Let $u,v\in V(H)$, and $P=\pi_{G-D}(u,v)$. If $\DecTree(u,v,D)>|P|$, then $P$ is not $\epsilon_3$-far away from $V(H)$.
	\end{lemma}
	
	\begin{figure}[H]
		\centering
		\includegraphics[width=0.6\linewidth]{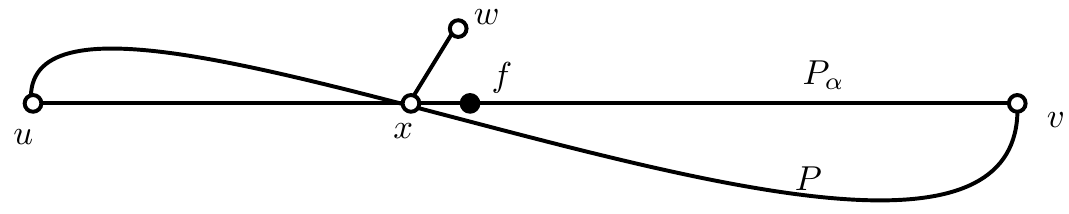}
	\end{figure}
	
	\begin{proof}
		For the sake of contradiction, assume $P$ is $\epsilon_3$-far away from $V(H)$. By \cref{thm:decomposable}, $P$ is an $\epsilon_4$-segment expath.
		
		Let $\alpha$ be the last decision tree node visited by $\DecTree(u,v,D)$ such that $\avoid(\alpha)\cap P=\varnothing$. Since $\DecTree(u,v,D)>|P|$, $P$ reaches some vertex $x\in\avoid(\alphanext)\setminus\avoid(\alpha)$, where $\alphanext$ is the next decision tree node visited by $\DecTree(u,v,D)$ after $\alpha$. Recall that $\avoid(\alphanext)\setminus\avoid(\alpha)=\seg(f,P_\alpha)\cap U_{l(f)}$, where $f$ is the failed vertex chosen in \cref{line:choose-f} of \cref{alg:DecTree}. By \cref{lemma:VH}, there is a vertex $w\in V(H)$ such that $|\pi_{G-D}(x,w)|\leq (\epsilon_3/2) \min\{|P_\alpha[u,x]|,|P_\alpha[x,v]|\}$.
		
		As $P_\alpha$ is the shortest $\epsilon_4$-expath from $u$ to $v$ in $G-\avoid(\alpha)$, and $P$ is \emph{some} such path, we have $|P|\ge |P_\alpha|$. We will prove $|P_\alpha[u,x]|\le 2|P[u,x]|$ and $|P_\alpha[x,v]|\le 2|P[x,v]|$, then it will follow that $|\pi_{G-D}(x,w)|\le \epsilon_3\min\{|P[u,x]|,|P[x,v]\}$, contradicting that $P$ is $\epsilon_3$-far away from $V(H)$. We only prove $|P_\alpha[u,x]|\le 2|P[u,x]|$, and the case that $|P_\alpha[x,v]|\le 2|P[x,v]|$ is symmetric.
		
		Suppose $|P[u,x]|<|P_\alpha[u,x]|/2$, we claim that the path $P[u,x]\circ P_\alpha[x,v]$ is a valid $\epsilon_4$-expath. Since $|P[u,x]|<|P_\alpha|/2\le |P|/2$, $x$ is closer to $u$ than to $v$ in $P$. Suppose $P_\alpha$ is composed of subpaths $P^\alpha_0,\dots,P^\alpha_{2B+1}$ interleaved with $\le 2B+3$ edges, and $P$ is composed of segments $P_1,\dots,P_\ell$. (Every $P^\alpha_j$ and $P_j$ is a shortest path in some $G_i$.) Recall from the proof of \cref{lemma:expath} that, if $x$ is in the first half of $P$, and $x\in P_k$, then $\lfloor\log_{1+\epsilon_4}|P[u,x]|\rfloor=k-1$.
		\begin{itemize}
			\item Let $x\in P^\alpha_j$, then $j\ge\lfloor\log_{1+\epsilon_4}|P_\alpha[u,x]|\rfloor$. This is because if $j<B+1$ (recall that $B=\lceil\log_{1+\epsilon_4}(nW)\rceil$ as in \cref{def:expath}), then $|P_\alpha[u,x]|\le (1+\epsilon_4)^j$.
			\item Let $x\in P_{j'}$, then $j'=\lfloor \log_{1+\epsilon_4}|P[u,x]|\rfloor+1$. Since $(1+\epsilon_4)^2<2\le |P_\alpha[u,x]|/|P[u,x]|$, we have $j'\leq \lfloor \log_{1+\epsilon_4}|P_\alpha[u,x]|\rfloor-1\leq j-1$.
		\end{itemize}
		
		\begin{figure}[H]
			\centering
			\includegraphics[width=0.8\linewidth]{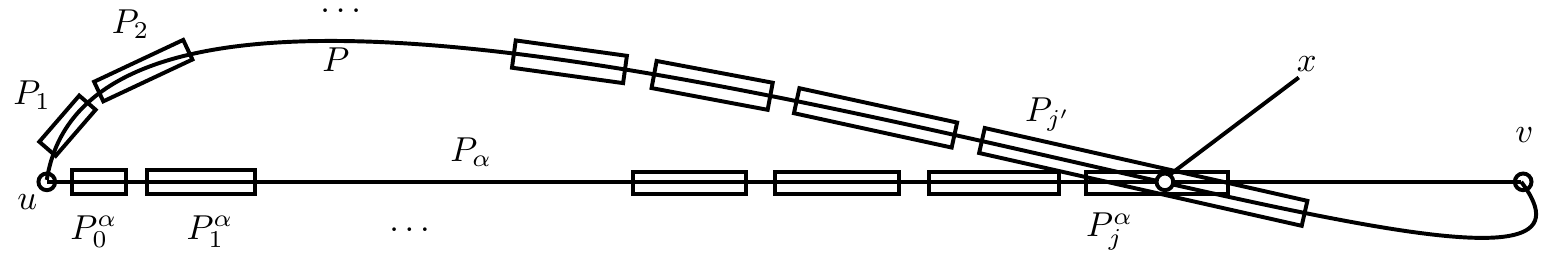}
		\end{figure}
		
		Let $P'=P[u,x]\circ P_{\alpha}[x,v]$, Consider the following representation of $P'$ as $P'_0,P'_1,\dots,P'_{2B+1}$:
		\begin{enumerate}[(i)]
			\item For $0\le i<j'$, $P'_i=P_i$.\label{item:prefixopt1}
			\item For $i=j'$, $P'_{j'}=P_{j'}[u_{j'},x)$, where $u_{j'}$ is the endpoint of $P_{j'}$ that lies on $P[u,x]$.\label{item:prefixopt2}
			\item For $j'< i< j$, $P'_i=\varnothing$.
			\item For $i=j$, $P'_j=P^{\alpha}_j[x,v^{\alpha}_j]$, where $v^{\alpha}_j$ is the endpoint of $P^{\alpha}_j$ that lies on $P_{\alpha}[x,v]$.\label{item:prefixopt3}
			\item For $j<i\le 2B+1$, $P'_i=P^{\alpha}_i$.\label{item:prefixopt4}
		\end{enumerate}
		
		\begin{figure}[H]
			\centering
			\includegraphics[width=0.8\linewidth]{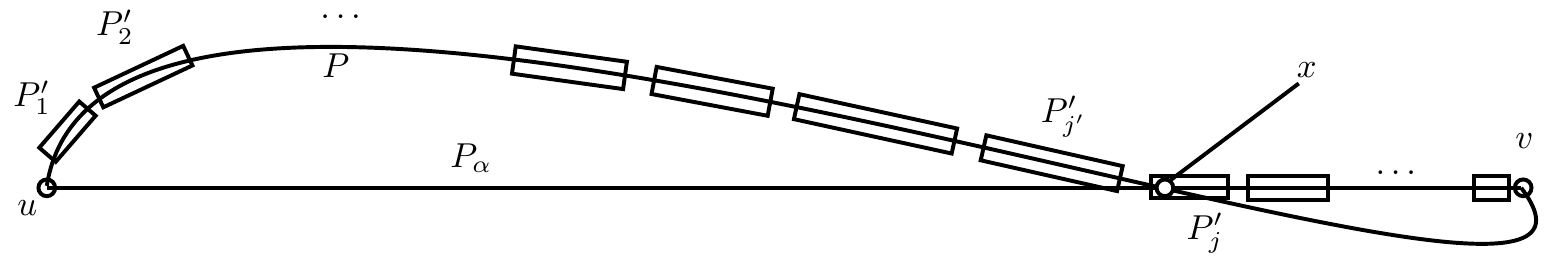}
		\end{figure}
		
		We need to verify that the representation $P'_0,P'_1,\dots,P'_{2B+1}$ satisfies the definition of $\epsilon_4$-expath. Let $u'_i,v'_i$ be the endpoints of $P'_i$, \ie $P'_i=P'[u'_i,v'_i]$, then:\begin{itemize}
			\item Case I: $i\le j'$ (\ie \cref{item:prefixopt1,item:prefixopt2}). In this case, $i<B+1$, as $P'_i$ lies in the first half of $P$. Since $|P'[u,v'_i]|=|P[u,v'_i]|\le (1+\epsilon_4)^i$, \cref{def:expath} is satisfied.
			\item Case II: $j\le i<B+1$. In this case, $|P'[u,v'_i]|=|P[u,x]|+|P_\alpha[x,v'_i]|<|P_\alpha[u,v'_i]|\le (1+\epsilon_4)^i$, thus \cref{def:expath} is satisfied.
			\item Case III: $j\ge B+1$. In this case, $|P'[u'_i,v]|=|P_\alpha[u'_i,v]|\le (1+\epsilon_4)^{2B+1-i}$, thus \cref{def:expath} is satisfied.
		\end{itemize}
		
		We conclude that $P'$ is a valid $\epsilon_4$-expath. Since $|P'|<|P_\alpha|$, this contradicts the choice of $P_{\alpha}$.
		
		Therefore $|P_\alpha[u,x]|\le 2|P[u,x]|$, and by symmetry, $|P_\alpha[x,v]|\le 2|P[x,v]|$. It follows that $P$ is not $\epsilon_3$-far away from $V(H)$.
	\end{proof}
	
	We prove the following theorem that immediately implies the approximation ratio of the algorithm.
	
	\begin{theorem}
		For every pair $u,v\in V(H)$, $|\pi_H(u,v)|\le (1+\epsilon)|\pi_{G-D}(u,v)|$.
	\end{theorem}
	\begin{proof}
		For the purpose of the proof, we construct a subgraph $H'$ of $H$ on the same set of vertices (\ie $V(H)$), but only keep the edges $(u,v)$ where $\pi_{G-D}(u,v)$ is $\epsilon_3$-far away from $V(H)$. By \cref{lemma:VH2}, the weight of every single edge $(u,v)$ in $H'$ is exactly $|\pi_{G-D}(u,v)|$. 
		
		We sort all pairs of vertices $u,v\in V(H)$ by nondecreasing order of $|\pi_{G-D}(u,v)|$. For every $u,v\in V(H)$, we define a $u$-$v$ path in $H'$ inductively in this order, and denote it as $p(u,v)$. The path $p(u,v)$ is defined as follows.
		\begin{itemize}
			\item If $\pi_{G-D}(u,v)$ is $\epsilon_3$-far away from $V(H)$, $p(u,v)$ consists of a single edge $(u,v)$.
			\item If $\pi_{G-D}(u,v)$ is not $\epsilon_3$-far away from $V(H)$, there exist $x\in \pi_{G-D}(u,v)\setminus \{u,v\}, w\in V(H)$ such that $|\pi_{G-D}(x,w)|\leq \epsilon_3\min\{|\pi_{G-D}(u,x)|,|\pi_{G-D}(x,v)|\}$. Since $\epsilon_3<1$, $|\pi_{G-D}(u,w)|$ and  $|\pi_{G-D}(w,v)|$ are both smaller than $|\pi_{G-D}(u,v)|$, so $p(u,w)$ and $p(w,v)$ are both well-defined. We concatenate these paths to form $p(u,v)$, \ie we define $p(u,v)=p(u,w)\circ p(w,v)$.
		\end{itemize}
		
		Let $k(u,v)$ be the number of edges in $p(u,v)$. We prove that for every $u,v\in V(H)$, 
		\[|\pi_{H'}(u,v)|\leq (1+\epsilon_3)^{k(u,v)}|\pi_{G-D}(u,v)|.\]
		
		We proceed by induction on $k(u,v)$. When $k(u,v)=1$, $|\pi_{H'}(u,v)|=|\pi_{G-D}(u,v)|$. Assume this is true for all pairs $(u,v)$ such that $k(u,v)<j$, consider some $(u,v)$ such that $k(u,v)=j$. Let $x,w$ be the vertices selected in the construction of $p(u,v)$, then both $k(u,w)$ and $k(w,v)$ are less than $j$. As $|\pi_{G-D}(x,w)|\le (\epsilon_3/2)|\pi_{G-D}(u,v)|$, we have
		\begin{align*}
			|\pi_{H'}(u,v)| \leq &~ |\pi_{H'}(u,w)|+|\pi_{H'}(w,v)| \\
			\leq &~ (1+\epsilon_3)^{j-1}(|\pi_{G-D}(u,w)|+|\pi_{G-D}(w,v)|) \\
			\leq &~ (1+\epsilon_3)^{j-1}(|\pi_{G-D}(u,v)|+2|\pi_{G-D}(x,w)|) \\
			\leq &~ (1+\epsilon_3)^{j-1}(|\pi_{G-D}(u,v)|+2\frac{\epsilon_3}{2}|\pi_{G-D}(u,v)|) \\
			\leq &~ (1+\epsilon_3)^j|\pi_{G-D}(u,v)|.
		\end{align*}
		Thus, for every $u,v\in V(H)$,
		\begin{align*}
			|\pi_{H}(u,v)|\leq &~ |\pi_{H'}(u,v)|\\
			\leq &~\left(1+\frac{\epsilon}{2|V(H)|}\right)^{|V(H)|}|\pi_{G-D}(u,v)|\\
			\leq &~ e^{\frac{\epsilon}{2}}|\pi_{G-D}(u,v)|\\
			< &~(1+\epsilon)|\pi_{G-D}(u,v)|.&\text{(since $\epsilon<1$)}&\qedhere
		\end{align*}
	\end{proof}
	
	Each $\epsilon_4$-expath can be stored in $O(\epsilon_4^{-1}\log(nW))$ space. Each non-leaf node in the decision tree has $O(h\epsilon_4^{-1}\log (nW))$ children. Thus we have a VSDO of
	\[\SpaceQueryStretch{n^{2+1/c}\epsilon_4^{-1}\log(nW)\cdot O(h\epsilon_4^{-1}\log(nW))^d}{O(d^2|V(H)|^2\cdot\epsilon_4^{-1}\log(nW))}{1+\epsilon.}\]
	As $\epsilon_4^{-1}=O(|V(H)|\cdot\epsilon^{-1}\log n)=O(d^{c+2}\epsilon^{-1}h\log^5 n)$, the VSDO is of
	\[\SpaceQueryStretch{n^{2+1/c}\cdot (\epsilon^{-1}d^c\log(nW))^{O(d)}}{\tilde{O}(\epsilon^{-1}d^{3c+8}\log W)}{1+\epsilon.}\]
	We improve both the space complexity and query time in the next subsection.
	
	\subsection{An Improvement}
	In \cref{sec:new_query}, we use $\epsilon_4$-segments in the decision tree. Therefore, each decision tree node that is not a leaf has $O(\epsilon_4^{-1}\cdot h\log(nW))$ children, and each decision tree node occupies $O(\epsilon_4^{-1}\cdot \log(nW))$ space. As $\epsilon_4^{-1}=\Theta(|V(H)|\cdot \epsilon^{-1}\log n)$, this $\epsilon_4^{-1}$ factor may seem too large. In this section, we show that the $|V(H)|$ factor in $\epsilon_4^{-1}$ can be shaved.
	
	Let $\epsilon_1=\epsilon/(2+\epsilon)$ as in \cref{sec:orac2} and $\epsilon_5=\epsilon_1/(4k-2)$. We will use $O(\epsilon_5^{-1}\log(nW))$ space to represent a node in the decision tree $FT(u,v)$. A first attempt would be to store the shortest $\epsilon_5$-expath in each node $\alpha$, but we face a technical problem as follows. Suppose $\DecTree(u,v,D)$ does \emph{not} capture the shortest path $P=\pi_{G-D}(u,v)$, then by \cref{lemma:pair}, $P$ is \emph{not} far from $V(H)$. In other words, there are vertices $x\in P$ and $w\in V(H)$ such that $\pi_{G-D}(x,w)\le \epsilon_1|P[u,x]|$. (Here we assume w.l.o.g.~that $x$ is closer to $u$.) Let $P_1=\pi_{G-D}(u,x)\circ\pi_{G-D}(x,w)$, and $P_2=\pi_{G-D}(w,v)$, we ``recursively'' find $P_1$ and $P_2$ and concatenate them as an approximation of $P$. The proof of \cref{lemma:pair} shows that $P_1$ is far away from $V(H)$, so we may attempt to use \cref{lemma:VH2} to conclude that $|\DecTree(u,w,D)|\le |P_1|$, and we only need to ``recurse'' on $P_2$. However, \cref{lemma:VH2} relies on \cref{thm:decomposable}, which requires $P_1$ to be a \emph{shortest} path in $G-D$, while $P_1=\pi_{G-D}(u,x)\circ\pi_{G-D}(x,w)$ is not necessarily the shortest $u$-$w$ path.
	
	The solution is simple. If $P_1$ is $\epsilon_1$-far from $V(H)$, we can use the same proof method of \cref{thm:decomposable}, to prove that each segment of $P_1=\pi_{G-D}(u,x)\circ\pi_{G-D}(x,w)$ is \emph{the concatenation of at most two} shortest paths in some $G_i$ and $G_j$. (The original \cref{thm:decomposable} proved that each segment of $\pi_{G-D}(u,v)$ is a shortest path in some $G_i$.) Therefore, we define \emph{segment bipaths}, in which each segment is the concatenation of two shortest paths in $G_i$ and $G_j$, rather than one shortest path in $G_i$ as in segment expaths.
	\begin{definition}
		A path $P$ in $G$ is an \emph{$\epsilon_5$-segment bipath} if the following holds. If we partition $P$ into $\epsilon_5$-segments as in \cref{def:segments}, for every segment $P[u_k,v_k]$, there exist two levels $i,j$ and a vertex $z\in P[u_k,v_k]$ such that $P[u_k,v_k]=\pi_{G_i}(u_k,z)\circ\pi_{G_j}(z,v_k)$.
	\end{definition}
	
	The following theorem can be proved by similar arguments as \cref{thm:decomposable}.
	\begin{theorem}\label{thm:decomposable2}
		For $u,v,w\in V$, let $P=\pi_{G-D}(u,v)\circ\pi_{G-D}(v,w)$. If $P$ is $\epsilon_1$-far away from $V(H)$, then it is an $\epsilon_5$-segment bipath.
	\end{theorem}
	
	\begin{figure}[H]
		\centering
		\includegraphics{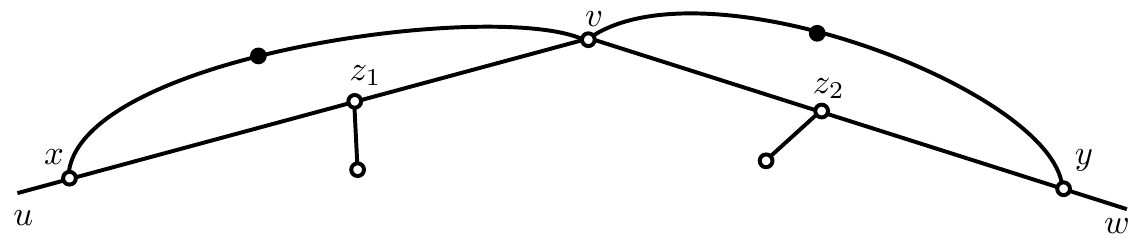}
	\end{figure}
	
	\begin{proof}[Proof Sketch]
		Let $P[x,y]$ be a segment of $P$. If $v\not\in P[x,y]$ then the argument of \cref{thm:decomposable} applies to $P[x,y]$, and there is some $1\le i\le p$ such that $P[x,y]=\pi_{G_i}(x,y)$. If $v\in P[x,y]$, then $P[x,v]$ and $P[v,y]$ are shortest paths in $G-D$ respectively. Let $z_1$ be the vertex with the highest level in $P[x,v]$, and $z_2$ be the vertex with the highest level in $P[v,y]$. We proceed with the same argument as in \cref{thm:decomposable}, and we can see that $P[x,v]$ is the shortest $x$-$v$ path in $G_{l(z_1)}$, and $P[v,y]$ is the shortest $v$-$y$ path in $G_{l(z_2)}$.
	\end{proof}
	
	Similarly we can define \emph{$\epsilon_5$-bipaths}:
	\begin{definition}\label{def:bipath}
		Let $B=\lceil\log_{1+\epsilon_5}(nW)\rceil$. An \emph{$\epsilon_5$-bipath} $P$ from $u$ to $v$ in $G$ is a path which is a concatenation of subpaths $P_0,\dots, P_{2B+1}$ interleaved with at most $2B+3$ edges, such that the following hold.\begin{itemize}
			\item For every $P_k=P[u_k,v_k]$ $(0\le k\le 2B+1)$, either $P_k$ is empty, or there exists a vertex $z\in P[u_k,v_k]$ and two levels $1\le i,j\le p$, such that $P[u_k,v_k]=\pi_{G_i}(u_k,z)\circ\pi_{G_j}(z,v_k)$.
			\item If $k<B+1$, then $|P[u,v_k]|\le(1+\epsilon_5)^k$; if $k\ge B+1$, then $|P[u_k,v]|\le(1+\epsilon_5)^{2B+1-k}$.
		\end{itemize}
	\end{definition}
	We also use $\epsilon_5$ in the definition of segments when constructing decision trees $FT(u,v)$. In each node $\alpha\in FT(u,v)$, we store the shortest $\epsilon_5$-bipath from $u$ to $v$ as the path $P_\alpha$. \cref{lemma:VH} still holds (for parameter $(2k-1)\epsilon_5=\epsilon_1/2$).
	
	\begin{reminder}{\cref{lemma:VH}}
		In the query algorithm $\DecTree(u,v,D)$, let $\alpha$ be a decision tree node it encounters, $f\in D$ be the failed vertex which is selected in \cref{line:choose-f} of \cref{alg:DecTree} and $i=l(f)$. (That is, $f$ is the vertex in $D\cap P_\alpha$ with the highest level.) For any non-failure vertex $x$ in $\seg(f,P_\alpha)\cap U_i$, there is a vertex $w\in V(H)$ such that $|\pi_{G-D}(x,w)|\leq (\epsilon_1/2) \min\{|P_\alpha[u,x]|,|P_\alpha[x,v]|\}$.
	\end{reminder}
	
	It is easy to verify that the counterparts of \cref{lemma:expath} and \cref{lemma:VH2} also hold for (segment) bipaths.
	\begin{lemma}\label{lemma:bipath}
		An $\epsilon_5$-segment bipath $P$ from $u$ to $v$ is an $\epsilon_5$-bipath.
	\end{lemma}
	\begin{lemma}\label{lemma:VH3}
		(Assume $1+\epsilon_5<\sqrt{2}$.) Let $y,w\in V(H)$, $x\in V$, and $P=\pi_{G-D}(y,x)\circ\pi_{G-D}(x,w)$. If $\DecTree(y,w,D)>|P|$, then $P$ is not $\epsilon_1$-far away from $V(H)$.
	\end{lemma}
	\begin{proof}[Proof Sketch of \cref{lemma:bipath} and \cref{lemma:VH3}]
		The arguments are essentially the same as \cref{lemma:expath,lemma:VH2}, except that each subpath in $P$ and $P_{\alpha}$ is now a concatenation of two shortest paths in $G_i$ and $G_{i'}$. This does not affect the calculation of lengths of paths in the proofs. In particular, in \cref{lemma:VH2}, the representation of $P[u,x]\circ P_{\alpha}[x,v]$ as $\epsilon_5$-bipath remains exactly the same, and it is easy to verify the validity of $P[u,x]\circ P_{\alpha}[x,v]$ as an $\epsilon_5$-bipath.
	\end{proof}
	
	Recall that the query algorithm builds the graph $H$ on vertex set $V(H)$, adds an edge of weight $\DecTree(x,y,D)$ for each $x,y\in V(H)$, and outputs the value $|\pi_H(u,v)|$. We now prove that the query algorithm has stretch $1+\epsilon$.
	\begin{theorem}
		For every $u,v\in V(H)$, $|\pi_H(u,v)|\le(1+\epsilon)|\pi_{G-D}(u,v)|$.
	\end{theorem}
	\begin{proof}
		For all pairs $u,v\in V(H)$, we sort the lengths $|\pi_{G-D}(u,v)|$ in nondecreasing order, and use induction on this order. For each $u,v\in V(H)$, if $\pi_{G-D}(u,v)$ is $\epsilon_1$-far away from $V(H)$, by \cref{lemma:VH3}, $\DecTree(u,v,D)=|\pi_{G-D}(u,v)|$ and we are done. Otherwise let $P=\pi_{G-D}(u,v)$, then there are vertices $x\in P\setminus\{u,v\}$, $y\in\{u,v\}$ and $w\in V(H)$ such that $|P[x,y]|\leq \frac{1}{2}|P|$ and $|\pi_{G-D}(x,w)|\le\epsilon_1|P[x,y]|$.
		
		Among all such triples $(x,y,w)$, we choose the triple that minimizes $|P[x,y]|$, and in case of ties choose the triple that minimizes $|\pi_{G-D}(x,w)|$. W.l.o.g.~assume $y=u$. Let $P'=\pi_{G-D}(u,x)\circ\pi_{G-D}(x,w)$, if $P'$ is not $\epsilon_1$-far away from $V(H)$, then there are vertices $x'\in P'\setminus\{u,w\},y'\in\{u,w\}$ and $w'\in V(H)$ such that $|\pi_{G-D}(x',w')|\le\epsilon_1|P'[x',y']|$. The same argument as \cref{lemma:pair} shows that this is a contradiction to the choice of $(x,y,w)$:\begin{itemize}
			\item If $x'\in P'[u,x)$, then the triple $(x',y,w')$ also satisfies that $|\pi_{G-D}(x',w')|\le \epsilon_1|P[x',y]|$, and $|P[x',y]|<|P[x,y]|$. So we should have chosen the triple $(x',y,w')$ instead of $(x,y,w)$.
			\item If $x'\in P'[x,w]$, then $|\pi_{G-D}(x,w')|\le |P'[x,x']|+|\pi_{G-D}(x',w')|\le |P'[x,x']|+\epsilon_1|P'[x',w]|<|\pi_{G-D}(x,w)|$ as $\epsilon_1<1$. So we should have chosen the triple $(x,y,w')$ instead of $(x,y,w)$.
		\end{itemize}
		
		\begin{figure}[H]
			\centering
			\includegraphics[scale=0.7]{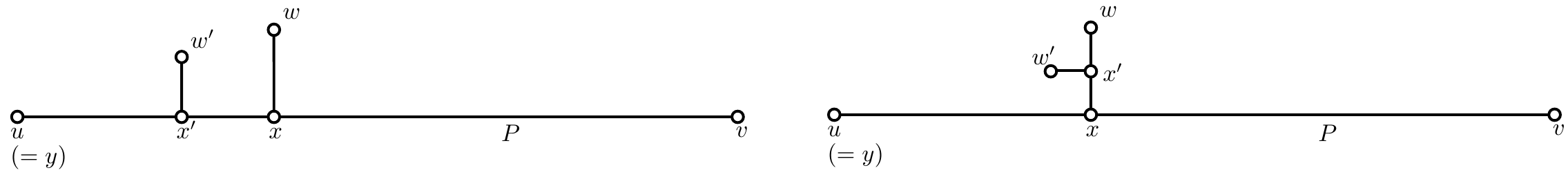}
		\end{figure}	
		
		It follows that $P'$ is $\epsilon_1$-far away from $V(H)$. By \cref{lemma:VH3}, we have $\DecTree(u,w,D)\le |P'|=|\pi_{G-D}(u,x)|+|\pi_{G-D}(x,w)|$. It is easy to see that $|\pi_{G-D}(w,v)|<|\pi_{G-D}(u,v)|$, thus by induction hypothesis $|\pi_H(w,v)|\le(1+\epsilon)|\pi_{G-D}(w,v)|$. We have
		\begin{align*}
			|\pi_H(u,v)|\le&~|\pi_H(u,w)|+|\pi_H(w,v)|\\
			\le&~|\pi_{G-D}(u,x)|+|\pi_{G-D}(x,w)|+(1+\epsilon)(|\pi_{G-D}(w,x)|+|\pi_{G-D}(x,v)|)\\
			\le&~|\pi_{G-D}(u,x)|(1+\epsilon_1+(1+\epsilon)\epsilon_1)+(1+\epsilon)|\pi_{G-D}(x,v)|\\
			=&~(1+\epsilon)|\pi_{G-D}(u,v)|.\qedhere
		\end{align*}
	\end{proof}
	
	Since an $\epsilon_5$-bipath occupies $O(\epsilon_5^{-1}\log(nW))$ space, and each non-leaf node has $O(h\epsilon_5^{-1}\log(nW))$ children, we have a VSDO of
	\[\SpaceQueryStretch{n^{2+1/c}\epsilon_5^{-1}\log(nW)\cdot O(h\epsilon_5^{-1}\log(nW))^d}{O(d^2|V(H)|^2\cdot\epsilon_5^{-1}\log(nW))}{1+\epsilon.}\]
	As $\epsilon_5^{-1}=O(\epsilon^{-1}\log n)$, the VSDO is of
	\[\SpaceQueryStretch{n^{2+1/c}\cdot (\log d/\log n)\cdot O(\epsilon^{-1}\log^2 n\log(nW)/\log d)^{d+1}}{O(\epsilon^{-1}d^{2c+6}\log^{11}n\log(nW)/\log^2d)}{1+\epsilon.}\]
	
	\subsection{Implementation Details}\label{sec:orac23}
	\paragraph{Preprocessing.} Given a subgraph $G'$ of $G$, vertices $s,t\in V$ and $\epsilon'>0$, we show that the shortest $\epsilon'$-expath from $s$ to $t$ in $G'$ can be computed in polynomial time.
	
	Let $\pi'_{G'}(s,t)$ be the shortest path of the form $\pi_{G_i}(s,t)$, where $1\leq i\leq p$ and $\pi_{G_i}(s,t)\subseteq G'$. (Note that $\pi'_{G'}(s, t)$ may not exist). First we compute $\pi'_{G'}(s,t)$ for all pairs of $s,t\in V$. Then let $\pi(s,t,j)$ be the shortest $s$-$t$ path $P$ in $G'$ such that the following hold.
	\begin{itemize}
		\item $P$ is the concatenation of subpaths $P_0,\dots,P_j$ interleaved with $\le j+1$ edges. Moreover, denote $P_k=P[u_k,v_k]$, where $u_k,v_k$ are endpoints of $P_k$ and $u_k$ is the one closer to $u$, then $v_j=t$, but there might be an edge between $s$ and $u_0$. (That is, $P$ is the concatenation of $e_0,P_0,e_1,P_1,\dots,e_j,P_j$ where each $e_i$ is an edge and each $P_i$ is a subpath.)
		\item For every $0\le k\le j$, $P_k$ is either empty or a shortest path in $G_i$ for some level $1\le i\le p$.
		\item For every $0\le k\le j$, $|P[s,v_k]|\le (1+\epsilon')^k$.
	\end{itemize}
	
	We use a dynamic programming algorithm to compute $|\pi(s,t,k)|$ for all $k\le B$. To start with, we artificially define $|\pi(s,t,-1)|$ as:
	\begin{equation*}
		|\pi(s,t,-1)|=\begin{cases}0&\text{if $s=t$}\\+\infty&\text{if $s\ne t$}\end{cases}.
	\end{equation*}
	Given $\{|\pi(s,t,j-1)|\}$ for all $s,t\in V$, we compute $|\pi(s,t,j)|$ as follows:
	\begin{equation*}
		|\pi(s,t,j)|=\begin{cases}|\tilde{\pi}(s,t,j)|&\text{if $|\tilde{\pi}(s,t,j)|\le(1+\epsilon')^j$}\\+\infty&\text{otherwise}\end{cases},
	\end{equation*}
	where
	\begin{equation}
		|\tilde{\pi}(s,t,j)|=\min_{(u,v)}\{|\pi(s,u,j-1)|+w(u,v)+|\pi'_{G'}(v,t)|\}.\label{eq:tildepi}
	\end{equation}
	Then the length of shortest $\epsilon'$-expath is
	\begin{equation*}
		\min_{(u,v)}\{|\pi(s,u,B)|+w(u,v)+|\pi(v,t,B)|\}.
	\end{equation*}
	Here $w(u,v)$ is the weight of the edge between $u$ and $v$. If $u=v$ then we assume $w(u,v)=0$.
	
	We can easily adapt the algorithm to obtain the actual shortest $\epsilon'$-expath.
	
	If we replace the term $\pi_{G'}'(s,t)$ in \eqref{eq:tildepi} by $\pi_{G'}''(s,t)$, which is defined as the shortest concatenated path of the form $\pi'_{G'}(s,u)\circ\pi'_{G'}(u,t)$, then we can also compute shortest $\epsilon'$-bipaths in polynomial time. Once we have a polynomial-time algorithm for computing the shortest $\epsilon'$-expath or $\epsilon'$-bipath in a subgraph $G'$, it is easy to see that the whole preprocessing time is polynomial in the space complexity.
	
	\paragraph{Query.} An $\epsilon'$-expath from $u$ to $v$ is stored as $O(\epsilon'^{-1}\log(nW))$ triples $(x,y,l)$, where each triple denotes a subpath $\pi_{G_l}(x,y)$. To check whether a failed vertex $f$ is in an $\epsilon'$-expath $P_\alpha$, we check every subpath $\pi_{G_l}(x,y)$ whether it contains $f$ by checking whether $\pi_{G_l}(x,f)+\pi_{G_l}(f,y)=\pi_{G_l}(x,y)$. The correctness of this method relies on the uniqueness assumption of shortest paths. If $f$ is in $P_{\alpha}$, we can also find the segment it is in, by computing $\lfloor\log_{1+\epsilon'}|P_\alpha[u,f]|\rfloor$ or $\lfloor\log_{1+\epsilon'}|P_\alpha[f,v]|\rfloor$.
	
	If we store the distance matrices of each $G_i$ during preprocessing, then every operation (\ie checking if $f\in P_\alpha$ and locating $\seg(f,P_\alpha)$) can be done in $O(\epsilon'^{-1}\log(nW))$ time. Therefore the time complexity of \cref{alg:DecTree} becomes $O(\epsilon'^{-1}\log(nW)\cdot d^2)$. Similar arguments also apply to $\epsilon'$-bipaths.
	
	\subparagraph{Retrieving the actual path.} The actual $(1+\epsilon)$-approximate shortest path can be efficiently retrieved as follows. (By \emph{retrieving a path efficiently}, we mean finding it in $O(\ell)$ additional time, where $\ell$ is the number of vertices in the path.)\begin{itemize}
		\item For every $1\le i\le p$, we also preprocess the shortest paths of $G_i$. That is, for every $v\in V(G_i)$, we precompute the incoming shortest path tree rooted at $v$. Consequently, given any $1\le i\le p$ and $u,v\in V(G_i)$, we can retrieve the path $\pi_{G_i}(u,v)$ efficiently.
		\item Let $\alpha\in FT(u,v)$ be a decision tree node. Recall that $P_\alpha$ is an $\epsilon_4$-expath or an $\epsilon_5$-bipath, therefore a concatenation of $O(\epsilon_4^{-1}\log(nW))$ or $O(\epsilon_5^{-1}\log(nW))$ paths of the form $\pi_{G_i}(x,y)$. Hence, $P_\alpha$ can be retrieved efficiently.
		\item Let $u,v\in V$ and $D$ be a set of failed vertices. We build the graph $H$ according to \cref{def:VH}, and find the shortest $u$-$v$ path in $H$. Each edge $(x,y)$ in this path corresponds to a path returned by $\DecTree(x,y,D)$, which by \cref{alg:DecTree} is $P_\alpha$ for some decision tree node $\alpha$. The concatenation of these paths $P_\alpha$ for each edge on $\pi_H(u,v)$ forms an $(1+\epsilon)$-approximate shortest $u$-$v$ path in $G-D$. As each $P_\alpha$ can be retrieved efficiently, this path can also be retrieved efficiently.
	\end{itemize}

	\subsection{A Reduction from Arbitrary Weights to Bounded Weights}\label{sec:arbitrary-to-bounded}
	If $W=n^{\omega(1)}$, then we may be unsatisfied with the $\log^d(nW)$ factor in the space complexity of our oracle. We can replace the $\log^d(nW)$ factor by $\log W\log^{d-1}n$ in the space complexity of our data structure, via a reduction from arbitrary weights to bounded weights. This reduction appears in \cite[Lemma 4.1]{CCFK17} and we notice that it also holds for vertex failures.
	\begin{lemma}[Lemma 4.1 of \cite{CCFK17}, rephrased]
		Suppose we have a VSDO for undirected graphs with edge weights in $[1,n^3]$, which occupies $S$ space, needs $Q$ query time and has stretch $A$. Then we can build a VSDO for undirected graphs with edge weights in $[1,W]$, which occupies $O(S\log W/\log n)$ space, needs $O(Q\log\log W)$ query time and has stretch $(1+1/n)A$. \label{lemma:bounded-weights}
	\end{lemma}
	\begin{proof}
		For every $0\le i\le\frac{\log W}{\log n}$, we build a VSDO $\caO^i$ on the graph $\tilde{G}^i$, which is defined as follows: $V(\tilde{G}^i)=V(G)$ and for each edge $(u,v)$ of weight $w$ in $G$, if $w\le n^{i+1}$, then we have an edge $(u,v)$ of weight $\lceil w\cdot n^{-(i-2)}\rceil$ in $\tilde{G}^i$. Note that the graphs $\tilde{G}^i$ are \emph{monotone} in the sense that, if an edge appears in $\tilde{G}^i$, then it also appears (albeit with a different weight) in $\tilde{G}^{i+1}$. Also note that the edge weights in every $\tilde{G}^{i+1}$ is at most $n^3$.
		
		Given a query $(u,v,D)$, we can use binary search to find the smallest integer $i$ such that $s$ and $t$ are connected in $\tilde{G}^i-D$. Then we use the oracle $\caO^i$ and $\caO^{i+1}$ to compute an $A$-approximation of the value
		\[ans=\min\mleft\{\delta_{\tilde{G}^i-D}(u,v)\cdot n^{i-2},\delta_{\tilde{G}^{i+1}-D}(u,v)\cdot n^{i-1}\mright\}.\]
		
		It remains to prove that $\delta_{G-D}(u,v)\le ans\le (1+1/n)\delta_{G-D}(u,v)$. That $\delta_{G-D}(u,v)\le ans$ is trivial. Let $\tilde{W}$ be the largest edge weight in $\pi_{G-D}(u,v)$, and $\istar=\lfloor\log_n\tilde{W}\rfloor$. Since $\tilde{W}\le n^{\istar+1}$, $u,v$ are connected in $\tilde{G}^\istar$. On the other hand, every edge in $G$ that appears in $\tilde{G}^{\istar-2}$ has weight at most $n^{\istar-1}$, thus if $u,v$ are connected in $\tilde{G}^{\istar-2}$, then $\delta_{G-D}(u,v)\le (n-1)n^{\istar-1}<n^\istar$, contradicting the definition of $\istar$. Therefore $\istar-1\le i\le \istar$ and $\istar\in\{i,i+1\}$.
		
		We have $ans\le \delta_{\tilde{G}^\istar-D}(u,v)\cdot n^{\istar-2}$. For every edge $e\in E(G)$ with weight $w$, if $e$ appears in the graph $\tilde{G}^\istar$, then $(\lceil w\cdot n^{-(\istar-2)}\rceil \cdot n^{\istar-2})\le w+n^{\istar-2}$, \ie every such edge is ``overestimated'' by an additive error of at most $n^{\istar-2}$. It follows that $ans\le\delta_{G-D}(u,v)+(n-1)n^{\istar-2}$. Since $\delta_{G-D}(u,v)\ge n^\istar$, we have $ans\le (1+1/n)\delta_{G-D}(u,v)$.
		
		As our new oracle computes an $A$-approximation of $ans$, its stretch is $(1+1/n)A$.
	\end{proof}
	
	Assuming $\epsilon>2/n$, \cref{lemma:bounded-weights} transforms the VSDO in \cref{sec:orac2} into a VSDO of
	\[\SpaceQueryStretch{n^{3+1/c}(\log W/\log n)\cdot O(\epsilon^{-1}\log^3 n/\log d)^d}{O(d^{2c+6}\log^{10}n\log\log W/\log^2 d)}{1+\epsilon,}\]
	and the VSDO in \cref{sec:reduce-space} into a VSDO of 
	\[\SpaceQueryStretch{n^{2+1/c}(\log W\log d/\log^2 n)\cdot O(\epsilon^{-1}\log^3 n/\log d)^{d+1}}{O(\epsilon^{-1}d^{2c+6}\log^{12}n\log\log W/\log^2 d)}{1+\epsilon.}\]
	
	\section{A $\poly(\log n,d)$-Stretch Oracle}\label{sec:orac1}
	We present an oracle of space complexity $n^{2+1/c}\poly(\log(nW),d)$ that achieves $\poly(\log n,d)$ stretch and $\poly(\log (nW),d)$ query time. We actually consider a decision version of our problem, namely:\begin{enumerate}[a)]
		\item It is given a parameter $\rho$.
		\item If $\delta_{G-D}(u,v)\le\rho$, the data structure outputs \textsc{Yes}.\label{item:poly-b}
		\item For some $A=\poly(\log n,d)$, if $\delta_{G-D}(u,v)>\rho\cdot A$, then the data structure outputs \textsc{No}.\label{item:poly-c}
	\end{enumerate}
	
	A standard binary search argument shows that if the above decision version can be solved in space $S$, query time $Q$ and stretch $A$, then there is a VSDO of size $O(S\log(nW))$, query time $O(Q\log\log(nW))$ and stretch $2A$. (See also \cite{CLPR12}.) Let $\caO_{\rho}$ be the oracle solving the decision version, and we build a VSDO $\caO^\star$ as follows. The new oracle $\caO^\star$ consists of $O(\log(nW))$ old oracles $\{\caO_{2^i}:0\le i\le\lceil\log_2(nW)\rceil\}$. For convenience we assume $\caO_{2^{-1}}$ always outputs \textsc{No} and $\caO_{2^{\lceil\log_2(nW)\rceil+1}}$ always outputs \textsc{Yes}.
	
	On a query $u,v,D$, the oracle $\caO^\star$ finds some $i$ ($0\le i\le\lceil\log_2(nW)\rceil+1$) such that $\caO_{2^{i-1}}$ outputs \textsc{No} and $\caO_{2^i}$ outputs \textsc{Yes}, and outputs $ans=A\cdot 2^i$. Such $i$ always exists and can be found in $O(\log\log(nW))$ oracle calls by binary search.\footnote{We maintain an interval $[l,r]$ such that on this query, $\caO_{2^l}$ outputs \textsc{No} and $\caO_{2^r}$ outputs \text{Yes}. Initially $l=-1$ and $r=\lceil\log_2(nW)\rceil+1$. In each iteration, let $m=\lfloor(l+r)/2\rfloor$, and we query $\caO_{2^m}$. If $\caO_{2^m}$ returns \textsc{No}, we set $l\gets m$, otherwise we set $r\gets m$. After $O(\log\log(nW))$ oracle calls, we have $r-l=1$ and we are done.} Since $\caO_{2^{i-1}}$ outputs \textsc{No}, we have $\delta_{G-D}(u,v)\ge 2^{i-1}$ by \ref{item:poly-b}), so $ans\le 2A\cdot\delta_{G-D}(u,v)$. Since $\caO_{2^i}$ outputs \textsc{Yes}, we have $\delta_{G-D}(u,v)\le A\cdot 2^i$ by \ref{item:poly-c}), so $ans\ge\delta_{G-D}(u,v)$. Thus $\caO^\star$ is indeed a VSDO with stretch $2A$.
	
	\subsection{Preliminaries}
	\subsubsection{$k$-Covering Sets}
	Denote $[l,r]=\{l,l+1,\dots,r\}$ and $[n]=[1,n]$. An \emph{interval} is a set of the form $[l,r]$. Given a universe $[n]$, a set of intervals $\caI$ is a \emph{$k$-covering set} of $[n]$ if for every $1\le l\le r\le n$, there are at most $k$ intervals $I_1,I_2,\dots,I_k\in\caI$ such that $\bigcup_{i=1}^kI_i=[l,r]$.
	
	The notion of $k$-covering sets arise from the study of the \emph{semigroup range query} problem \cite{Yao82, AS87}, which is a generalization of the \emph{range minimum query} problem \cite{AHU73,HT84,BF00}. For example, by constructing an interval tree over $[n]$, it is easy to see that there is an $O(\log n)$-covering set of $[n]$ whose size is $O(n)$. We use the following (stronger) results of \cite{Yao82, AS87}:
	\begin{theorem}
		There exists a polynomial-time computable $O(\alpha(n))$-covering set $\caI$ of $[n]$ with $|\caI|=O(n)$, where $\alpha(n)$ is the inverse-Ackermann function. Moreover, for any interval $[l,r]$, we can find $O(\alpha(n))$ intervals whose union is $[l,r]$ in $O(\alpha(n))$ time.\label{thm:covering-set}
	\end{theorem}
	\subsubsection{Euler Tours}
	For a tree $T$ rooted at $r\in V$, we perform a depth-first search on $T$ starting at $r$, and record every vertex at the first time it is encountered. The sequence of encountered vertices is called the \emph{Euler tour} of $T$, denoted as $\ET(T)$. The Euler tour has a nice property, namely that every subtree of $T$ rooted at $x\in V(T)$ corresponds to an interval of the sequence $\ET(T)$. As a corollary, if we remove $d$ vertices $D$ from $T$, every connected component in $T-D$ (which is a smaller tree) corresponds to the union of $O(d)$ such intervals. (See \cref{fig:ET} for an illustration.)
	
	\begin{lemma}
		Let $D$ be a subset of $V(T)$ such that $|D|\le d$, $S$ be any connected component of $T-D$, then $S$ is the union of $O(d)$ intervals of $\ET(T)$. Moreover these intervals can be found in $O(d\log d)$ time.\label{lemma:ET}
	\end{lemma}
	\begin{proof}
		Let $I(x)$ denote the interval of $\ET(T)$ corresponding to the subtree rooted at $x$. Let $h$ be the highest vertex in $S$, then $S=I(h)\setminus\bigcup_{f\in D}I(f)$, which is a big interval subtracting $d$ smaller intervals. By sorting the endpoints of $\{I(f):f\in D\cup\{h\}\}$, we can express $S$ as the union of $O(d)$ intervals.
	\end{proof}
	
	\begin{figure}
		\centering
		\includegraphics[scale=1]{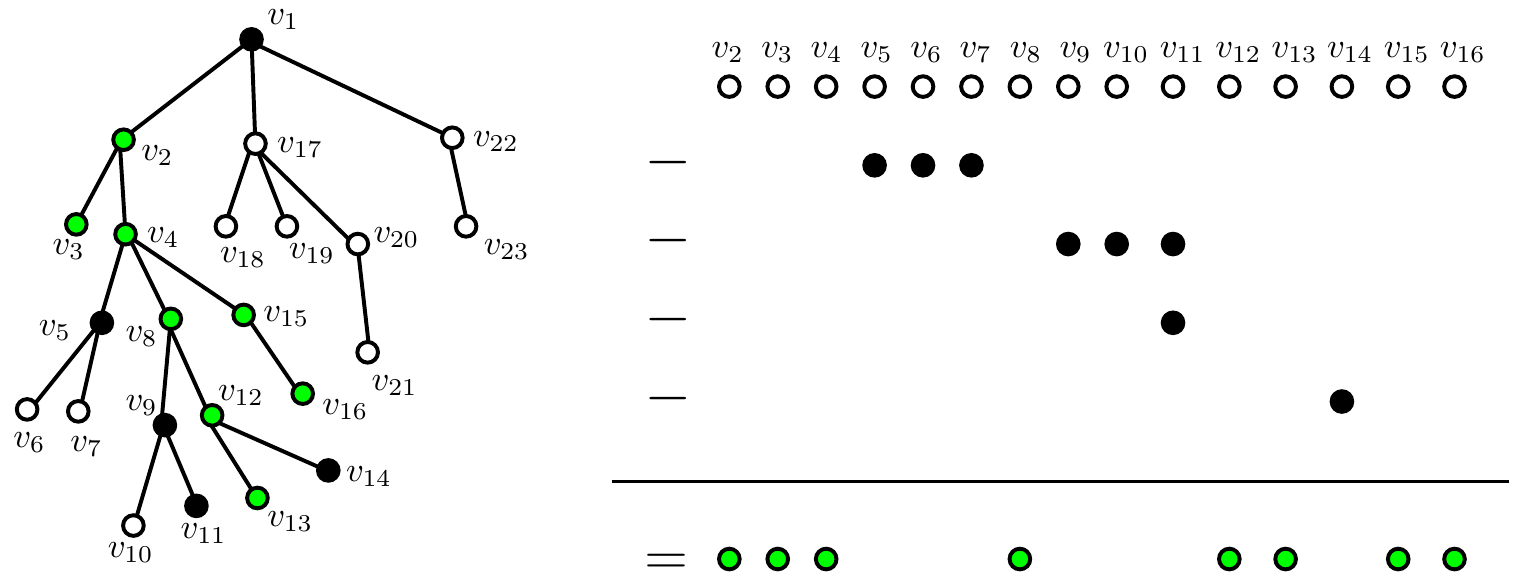}
		\caption{Left: a sample tree $T$. A possible Euler tour of $T$ is $v_1,v_2,\dots,v_{23}$. We delete the vertices $\{v_1,v_5,v_9,v_{11},v_{14}\}$ from $T$, and the connected component containing $v_2$ is marked as green. Right: the green component is a big interval (corresponding to the subtree of $v_2$) subtracting $4$ smaller intervals (subtrees of $v_5,v_9,v_{11}$ and $v_{14}$), thus the union of $\le 5$ intervals.}\label{fig:ET}
	\end{figure}
	
	\subsection{Preprocessing Algorithm}\label{subsec:pre-orac1}
	In the preprocessing algorithm, (for each path $V=U_1,U_2,\dots,U_p$ of the hierarchy tree,) we prune the trees and construct auxiliary data structures $H,E'$ as follows.
	
	\paragraph{Pruning the trees.} Recall that we consider the trees in $\caT=\bigcup_{i=1}^p\caT_{U_{i+1}}(U_i)$, and we have a distance parameter $\rho$. We prune off vertices of large depth in every tree in $\caT$. For every $T\in\caT$ and $v\in V(T)$, if $\dep_T(v)>(2k-1)\rho$, we delete $v$ from $T$. There are two reasons to perform this step:\begin{itemize}
		\item If for some $u\in V$, $\delta(u,v)\le\rho$ and $u,v$ are covered by $T$ (\ie $\dep_T(u)+\dep_T(v)\le (2k-1)\cdot \delta(u,v)$), then $\dep_T(v)\le(2k-1)\rho$, hence the pruning would not affect any distance of $\le\rho$;
		\item After the pruning, every tree in $\caT$ has diameter at most $(4k-2)\rho$.
	\end{itemize}
	
	In the rest of this section, we assume that all trees in $\caT$ are pruned.
	
	\paragraph{The auxiliary DAG $H$.} We list the trees as $\caT=\{T_1,T_2,\dots,T_{|\caT|}\}$, and concatenate their Euler tours as a list $\Lambda=\ET(T_1)\circ\ET(T_2)\circ\dots\circ\ET(T_{|\caT|})$. Recall that $|\caT|=O(n)$ and every vertex appears in $O(h\log^2 n)$ trees, where $h=O(\log n/\log d)$. Therefore $|\Lambda|=O(nh\log^2 n)$ and $\alpha(|\Lambda|)=O(\alpha(n))$. Let $\caI$ be an $O(\alpha(n))$-covering set of $\Lambda$, so every interval of $\Lambda$ can be expressed as the union of $O(\alpha(n))$ intervals in $\caI$. We make two copies $\caI_1,\caI_2$ of $\caI$, two copies $\caT_1,\caT_2$ of $\caT$, and one copy $V_1$ of $V$. For $I\in\caI$, let $I_1,I_2$ be its copies in $\caI_1,\caI_2$ respectively; $T_1,T_2,v_1$ are similarly defined.
	
	We define a DAG $H$ with $V(H)=\caI_1\cup \caT_1\cup V_1\cup \caT_2\cup \caI_2$, and $E(H)$ defined as follows (where $(a\to b)$ denotes a directed edge from $a$ to $b$):
	\begin{enumerate}
		\item Let $I\in\caI$, $T\in \caT$, if there is an edge from some vertex in $I$ to some vertex in $T$ with weight $\le\rho$, then we have edges $(I_1\to T_1)$ and $(T_2\to I_2)$ in $E(H)$;
		\item Let $T\in\caT, u\in V(T)$ (after the pruning), then we have edges $(T_1\to u_1)$ and $(u_1\to T_2)$ in $E(H)$.
	\end{enumerate}
	
	In the query algorithm, we use the graph $H$ to capture the paths only ``involved'' with \emph{unaffected trees}, which are trees that do not intersect $D$ (see \cref{def:affected-tree}). Therefore, we need to remove $\caT^*_1\cup\caT^*_2$ from $H$, where $\caT^*$ is the set of affected trees, and $\caT^*_1,\caT^*_2$ are copies of $\caT^*$ in $\caT_1,\caT_2$ respectively. Suppose we can upper bound the number of affected trees as $|\caT^*|\le K$. We are interested in the following kind of queries on $H$: ``Given $I_1\in\caI_1,I_2\in\caI_2$, can $I_1$ reach $I_2$ in $H-\caT^*_1-\caT^*_2$?'' We claim that, since the depth of $H$ is a constant, such queries can be answered efficiently.
	\begin{lemma}\label{lemma:bounded-depth}
		Let $H=(V,E)$ be a DAG, $V'\subseteq V$, $n=|V|$, $x,y$ be integers and $q=\frac{(x+y)^{x+y+1}}{x^xy^y}\log n$. Suppose every path in $H$ contains at most $x$ vertices in $V'$. We can build a data structure of size $O(qn^2)$ which, given a subset $D\in\binom{V'}{\le y}$ and two vertices $u,v\in V$, in $O(q)$ query time, outputs $1$ if $u$ can reach $v$ in $H-D$ and $0$ otherwise. With high probability over the randomized preprocessing algorithm, the data structure is correct on every query.
	\end{lemma}
	
	The proof uses a clever trick of \cite{DK11}, which was inspired by color-coding \cite{AYZ95}.
	
	Let $\caU=\{1,2,\dots,n\}$, $\caS$ be a family of subsets of $\caU$. We say $\caS$ is an \emph{$(x,y)$-family} if for every $X\in\binom{\caU}{x}$, $Y\in\binom{\caU}{y}$ such that $X\cap Y=\varnothing$, there is a set $S\in\caS$ such that $X\subseteq S$ and $Y\cap S=\varnothing$.
	
	Fix $X,Y$, we randomly sample a subset $S$ of $\caU$ by picking every element w.p.~$\frac{x}{x+y}$. Then the set $S$ satisfies the condition that $X\subset S$ and $Y\cap S=\varnothing$ w.p.~$p=\frac{x^xy^y}{(x+y)^{x+y}}$. By a union bound over all $X,Y$'s, if we sample $O(p^{-1}\log(n^{x+y}))=O\mleft(\frac{(x+y)^{x+y+1}}{x^xy^y}\log n\mright)$ such sets $S$, we obtain an $(x,y)$-family with high probability.
	
	\begin{remark}
		The above construction of $(x, y)$-family can be derandomized by \cite[Theorem 14]{KP20}. In our regime where $x = 2$ and $y = {\rm polylog}(n)$, the randomized construction contains $O(y^3\log n)$ sets, while the deterministic construction contains $O(y^3\log^3 n)$ sets, slightly worse than the randomized construction. Below we will still use the randomized construction.
	\end{remark}
	
	\begin{proof}[Proof of \cref{lemma:bounded-depth}]
		Let $\caS$ be an $(x,y)$-family of $V'$. For every $S\in\caS$, we store a reachability matrix of the induced subgraph $H[S\cup (V\setminus V')]$. On a query $(u,v,D)$, the algorithm outputs $1$ if and only if there is some $S\in\caS$ such that $S\cap D=\varnothing$ and $u$ can reach $v$ in $H[S\cup(V\setminus V')]$.
		
		The correctness of the algorithm follows directly from the definition of $(x,y)$-family. If $u$ does not reach $v$ in $H-D$, then for any $S$ such that $S\cap D=\varnothing$, $u$ does not reach $v$ in $H[S\cup (V\setminus V')]$. If $u$ can reach $v$ in $H-D$, let $P$ be the vertices in some specific path from $u$ to $v$ in $H-D$, then $|P\cap V'|\le x$ by hypothesis. By the definition of $(x,y)$-family, there is a set $S_0\in\caS$ such that $P\cap V'\subseteq S_0$ and $D\cap S_0=\varnothing$, and the algorithm detects that $u$ can reach $v$ in $G[S_0\cup(V\setminus V')]$.
	\end{proof}
	
	Let $x=2,y=2K,V'=\caT_1\cup\caT_2$ and $q=\frac{(x+y)^{x+y+1}}{x^xy^y}\log |\Lambda|=O(K^3\log n)$, \cref{lemma:bounded-depth} implies that we can maintain $H$ in $O(q\cdot |\Lambda|^2)$ space and answer the above queries in $O(q)$ time.
	
	\paragraph{The auxiliary table $E'$.} Besides the main structure $H$, we also need to store a table $E'$, specified as a subset of $\Lambda\times\Lambda$. For every $u,v\in V$ such that $u=v$ or $(u,v)$ is an edge with weight $\le\rho$ in $E$, for every occurrences $u',v'$ of $u,v$ in $\Lambda$ respectively, there is an item $(u',v')\in E'$. There are no other items in $E'$.
	
	Since every vertex occurs $O(h\log^2 n)$ times in $\Lambda$, we have $|E'|=O(mh^2\log^4 n)$. We store $E'$ by a 2D range search structure \cite{ABR00} of size $O(mh^2\log^5 n)$ such that given intervals $I_1,I_2$ of $\Lambda$, it can be queried if $E'\cap (I_1\times I_2)=\varnothing$ in $O(\log \log n)$ time.
	
	\subsection{Query Algorithm}
	Suppose we are given $u,v\in V$, $D\in\binom{V}{\le d}$ and $\rho$. As described earlier, we have already found a path $V=U_1,U_2,\dots,U_p$ in the hierarchy tree where every vertex in $D$ has low pseudo-degree in every tree.
	
	\paragraph{Identifying affected trees.} We first identify the \emph{affected trees} in $\caT$. After the removal of $D$, these trees split into several subtrees, some of which are called \emph{affected subtrees}, and the others are \emph{ignored subtrees}. A precise definition is as follows:
	\begin{definition}\label{def:affected-tree}
		A tree $T\in\caT$ is an \emph{affected tree} if $V(T)\cap D\ne\varnothing$. For each affected tree $T\in\caT$, the removal of $D$ splits $T$ into several subtrees $T^{(1)},T^{(2)},\dots,T^{(q)}$. A subtree $T^{(i)}$ is an \emph{affected subtree} if it contains some trunk vertex of $T$; otherwise it is an \emph{ignored subtree}. A vertex $v\in V$ is an \emph{affected vertex} if it is in some affected subtree; otherwise it is an \emph{unaffected vertex}.
	\end{definition}
	\begin{remark}
		It is possible that an unaffected vertex belongs to some ignored subtree.
	\end{remark}
	
	\begin{lemma}[The Number of Affected (Sub)Trees]
		There are at most $O(dh\log^2 n)$ affected trees and $O(d^{c+2}h\log^4 n)$ affected subtrees.
	\end{lemma}
	\begin{proof}
		Recall that $p=O(h)$ is the depth of the high-degree hierarchy. Every vertex in $D$ is in at most $p\cdot 2e\ln^2 n$ trees, so at most $O(dh\log^2 n)$ trees can be affected. Every vertex in $D$ has pseudo-degree at most $s=O(d^{c+1}\log^2 n)$ in every tree, so these affected trees split into at most $O(d^{c+2}h\log^4 n)$ affected subtrees (and possibly many ignored subtrees).
	\end{proof}
	
	\paragraph{The graph $R$.} During the query algorithm, we construct an unweighted graph $R$ whose vertex set is $V(R)=\{\{u\},\{v\}\}\cup\caT'$, where $\caT'$ is the set of affected subtrees. We output \textsc{Yes} if and only if $\{u\}$ and $\{v\}$ are connected in $R$.
	
	For every $X,Y\in V(R)$, we use the following procedure to determine if there is an edge between $X$ and $Y$ in $E(R)$. We consider $X,Y$ as subsets of $V$. If $X$ is an affected subtree which belongs to the affected tree $T$, then by \cref{lemma:ET}, we can write $X$ as the union of $O(d)$ intervals of $\ET(T)$. By \cref{thm:covering-set}, $X$ is the union of $O(d\cdot\alpha(n))$ intervals in $\caI$. If $X=\{u\}$ or $X=\{v\}$ then $X$ is trivially an interval in $\caI$. Similarly, we can also represent $Y$ as the union of $O(d\cdot\alpha(n))$ intervals in $\caI$. If there are two intervals $I^1,I^2\in\caI$, where $I^1$ is in the representation of $X$ and $I^2$ is in the representation of $Y$, such that either $I_1^1$ can reach $I_2^2$ in $H-\caT^*_1-\caT^*_2$ or $E'\cap(I^1\times I^2)\ne\varnothing$, then we insert an edge in $E(R)$ between $X$ and $Y$. Here $\caT^*_1,\caT^*_2$ denote the copies of affected trees in $\caT_1,\caT_2$ of $V(H)$ respectively.
	
	The time complexity for the query algorithm is dominated by constructing $E(R)$. Since $|V(R)|=O(d^{c+2}h\log^4n)$, and there are at most $K=O(dh\log^2 n)$ affected trees, the algorithm takes $O((d^{c+2}h\log^4n)^2(d\cdot\alpha(n))^2(K^3\log n+\log\log n))=O(d^{2c+9}\alpha^2(n)h^5\log^{15}n)$ time.
	
	\paragraph{Justification.} We justify the construction of the graph $R$. For $X,Y\in V(R)$, there should be an edge between $X$ and $Y$ if there is an \emph{unaffected path} of length at most $\rho$ connecting them, defined as follows.
	\begin{definition}\label{def:unaffected-path}
		For $X,Y\subseteq V$, an \emph{unaffected path} in $G-D$ connecting $X$ and $Y$ is a path $(v_0,\dots,v_{\ell})$ $(\ell\ge 0)$ in $G-D$ such that $v_0\in X,v_{\ell}\in Y$, and $v_1,v_2,\dots,v_{\ell-1}$ are unaffected vertices.
	\end{definition}
	
	The following theorem justifies the definition of $E(R)$.
	
	\begin{theorem}\label{thm:E(R)}
		For $X,Y\in V(R)$:\begin{enumerate}[a)]
			\item If there is an unaffected path of length $\le\rho$ connecting $X$ and $Y$, then $(X,Y)\in E(R)$.\label{item:E(R)-a}
			\item If $(X,Y)\in E(R)$, then there is a path in $G-D$, which starts at some vertex in $X$, ends at some vertex in $Y$, and has length at most $(8k-2)\rho$.\label{item:E(R)-b}
		\end{enumerate}
	\end{theorem}
	\begin{proof}
		\textbf{Proof of \ref{item:E(R)-a}).} Let the path be $(v_0,v_1,\dots,v_{\ell})$ where $v_0\in I^1, v_{\ell}\in I^2$, and $I^1,I^2$ are intervals in the representation of $X,Y$ respectively. If $\ell\le 1$, then we have $(I^1\times I^2)\cap E'\ne\varnothing$. If $\ell>1$, let $v_j$ be the vertex in $v_1,v_2,\dots,v_{\ell-1}$ with the highest level, and $i=l(v_j)$ be its level. Then the path does not intersect $U_{i+1}$. Let $T^1=T_i(v_j,v_1)$ be the tree in $\caT_{U_{i+1}}(U_i)$ which approximates the distance $\delta_{G-U_{i+1}}(v_j,v_1)$, then $\dep_{T^1}(v_1)+\dep_{T^1}(v_j)\le(2k-1)\rho$. Therefore $v_1$ and $v_j$ are not pruned in $T^1$. If $T^1$ is an affected tree, then since $v_j$ is an unaffected vertex, it must lie in some ignored subtree of $T^1$, but this contradicts the fact that $v_j\in\Trunk(T^1)$. Therefore $T^1$ is not an affected tree. Similarly let $T^2=T_i(v_j,v_{\ell-1})$, then $\dep_{T^2}(v_{\ell-1})+\dep_{T^2}(v_j)\le(2k-1)\rho$, and $T^2$ is not an affected tree. We conclude that there is a path $I^1_1\to T^1_1\to (v_j)_1\to T^2_2\to I^2_2$ in $H-\caT^*_1-\caT^*_2$.
		
		\begin{figure}
			\centering
			\begin{minipage}[b]{0.49\textwidth}
				\centering
				\includegraphics[scale=1]{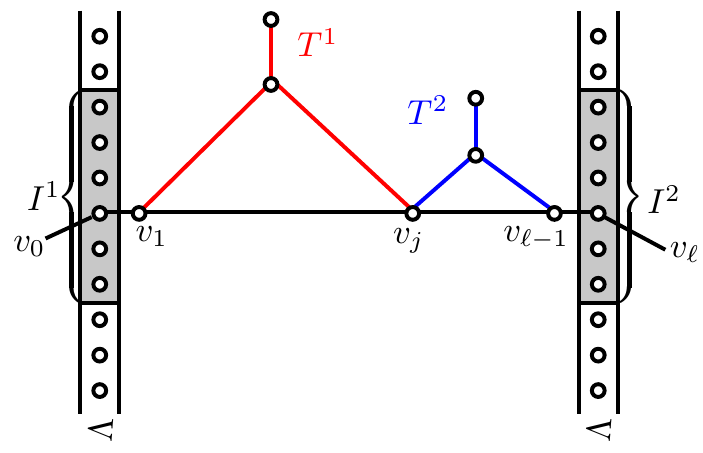}
				\caption{Proof of \ref{item:E(R)-a}).}
			\end{minipage}
			\begin{minipage}[b]{0.49\textwidth}
				\centering
				\includegraphics[scale=1]{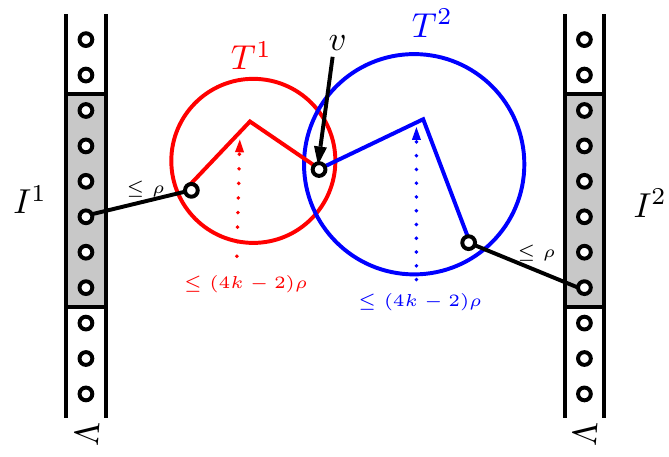}
				\caption{Proof of \ref{item:E(R)-b}).}
			\end{minipage}
		\end{figure}
		
		\textbf{Proof of \ref{item:E(R)-b}).} Suppose that the interval $I^1$ is in the representation of $X$, the interval $I^2$ is in the representation of $Y$, and $I^1,I^2$ contributes to the edge $(X,Y)$. If $(I^1\times I^2)\cap E'\ne\varnothing$, then either $I^1\cap I^2\ne\varnothing$ or there is an edge $(u,v)\in E$ of length $\le\rho$ such that $u\in I^1$ and $v\in I^2$. In either case the lemma follows. On the other hand, if $I^1_1$ can reach $I^2_2$ in $H-\caT^*_1-\caT^*_2$, and the corresponding path in $H-\caT^*_1-\caT^*_2$ is $I^1_1\to T^1_1\to v_1\to T^2_2\to I^2_2$, then $T^1$ and $T^2$ are unaffected trees. Consider the following path $p$, which starts from $I^1$, goes to an adjacent vertex in $T^1$ by an edge of weight $\le\rho$, walks along $T^1$ to reach $v$, walks along $T^2$ to reach a vertex adjacent to $I^2$ by an edge of weight $\le\rho$, then goes to $I^2$. Since every tree has diameter at most $(4k-2)\rho$, the length of $p$ is at most $(8k-2)\rho$. Since $T^1$ and $T^2$ are unaffected trees, $p$ avoids $D$, and the lemma follows.
	\end{proof}
	
	Given \cref{thm:E(R)}, it is easy to prove that our algorithm achieves a stretch of $O(k\cdot|V(R)|)=O(d^{c+2}h\log^5 n)$.
	
	\begin{theorem}[Correctness]\label{thm:correct1}
		There is some $A=O(d^{c+2}h\log^5n)$ such that for $u,v\in V$, $D\in\binom{V}{\le d}$:\begin{enumerate}[a)]
			\item If $\delta_{G-D}(u,v)\le\rho$, then the algorithm outputs \textsc{Yes}.\label{item:c1-a}
			\item If the algorithm outputs \textsc{Yes}, then $\delta_{G-D}(u,v)\le\rho\cdot A$.\label{item:c1-b}
		\end{enumerate}
	\end{theorem}
	\begin{proof}
		\textbf{Proof of \ref{item:c1-a}).} Suppose $p:(u=w_0,w_1,\dots,w_{\ell-1},w_{\ell}=v)$ is a path from $u$ to $v$ in $G-D$ with length at most $\rho$. For $1\le i<\ell$, let $t_i$ be any affected subtree that $w_i$ lies in; if $w_i$ is an unaffected vertex then set $t_i=\varnothing$. Let $a(i)$ $(i\ge 1)$ be the $i$-th index such that $t_{a(i)}\ne\varnothing$, and $q$ be the maximum index such that $a(q)$ is defined. (For example, $t_j=\varnothing$ for every $1\le j<a(1)$ or $a(q)<j<\ell$, but $t_{a(1)}\ne\varnothing,t_{a(q)}\ne\varnothing$.) Artificially we define $t_0=\{u\},a(0)=0,t_{\ell}=\{v\}$ and $a(q+1)=\ell$. Then for every $0\le i\le q$, we have an unaffected path $(w_{a(i)},w_{a(i)+1},\dots,w_{a(i+1)})$ of length $\le\rho$ connecting $t_{a(i)}$ and $t_{a(i+1)}$, thus $(t_{a(i)},t_{a(i+1)})\in E(R)$ by \cref{thm:E(R)} \ref{item:E(R)-a}). We conclude that $\{u\}$ and $\{v\}$ are connected in $R$. Therefore the algorithm returns \textsc{Yes}.
		
		\textbf{Proof of \ref{item:c1-b}).} Suppose the algorithm returns \textsc{Yes}. Then there is a simple path $\{u\}=t_0\to t_1\to\dots\to t_{\ell}=\{v\}$ in $R$. By \cref{thm:E(R)} \ref{item:E(R)-b}), there are vertices $u=u_0,v_1,u_1,\dots,v_{\ell-1},u_{\ell-1},v_{\ell}=v$ such that:\begin{itemize}
			\item For every $1\le i<\ell$, $v_i,u_i\in t_i$.
			\item For every $0\le i<\ell$, $\delta_{G-D}(u_i,v_{i+1})\le(8k-2)\rho$.
		\end{itemize}
		Since each tree has diameter at most $(4k-2)\rho$, we can add that:\begin{itemize}
			\item For every $1\le i<\ell$, $\delta_{G-D}(v_i,u_i)\le (4k-2)\rho$.
		\end{itemize}
		Therefore $\delta_{G-D}(u,v)\le \ell\cdot(8k-2)\rho+(\ell-1)\cdot(4k-2)\rho=\left((12k-4)\ell-4k+2\right)\rho$. Since $\ell\le|V(R)|=O(d^{c+2}h\log^4n)$, we have $\delta_{G-D}(u,v)\le O(d^{c+2}h\log^5n)\cdot\rho$.
	\end{proof}

	\begin{remark}
		By investigating the proofs of \cref{thm:E(R)} \ref{item:E(R)-b}) and \cref{thm:correct1} \ref{item:c1-b}), we can retrieve a path from $u$ to $v$ in $G-D$ of length $O(d^{c+2}h\log^5 n)\cdot \rho$ in $O(\ell)$ additional time, where $\ell$ is the number of nodes in the retrieved path.
	\end{remark}

	
	
	The space complexity of our oracle is dominated by the $O(n^{1/c}q|\Lambda|^2)$ term, therefore our VSDO has
	\[\SpaceQueryStretch{O(n^{2+1/c}d^3\log^{16}n\log(nW)/\log^5d)}{O(d^{2c+9}\alpha^2(n)\log^{20}n\log\log(nW)/\log^5d)}{O(d^{c+2}\log^6n/\log d).}\]

	\section*{Acknowledgments}
	We are grateful to anonymous reviewers for helpful comments, bringing \cite{RTZ05} to our attention, and pointing out the recent work \cite{KP20} that allows us to derandomize the oracle in \cref{sec:orac1}. We would like to thank Thatchaphol Saranurak for providing an early manuscript of \cite{BS19}, and Zhijun Zhang for helpful comments on a draft version of this paper.

	\bibliography{article}

\appendix
	
	\section{Proof of \cref{thm:sr-tree-cover}}
	\label{sec:proof-tree-cover}
	We start with a randomized construction. We construct a sequence of nested subsets $S=A_0\supseteq A_1\supseteq\dots\supseteq A_k=\emptyset$, where each $A_i$ ($1\le i<k$) is constructed by independently sampling each vertex in $A_{i-1}$ w.p.~$n^{-1/k}$. Let $w\in S$, then there is some $0\le i<k$ such that $w\in A_i\setminus A_{i+1}$. We define a \emph{cluster} $C(w)$ around $w$ as follows:
	\[C(w)=\mleft\{v\in V:\delta(v,w)<\delta(v,A_{i+1})\mright\}.\]
	That is, if $v$ is closer to $w$ than to every vertex in $A_{i+1}$, then $v\in C(w)$. Let $T(w)$ be the shortest path tree rooted at $w$ spanning $C(w)$. The tree cover is $\caT(S)=\{T(w):w\in S\}$.
	
	It is easy to see that the tree $T(w)$ only contains vertices in $C(w)$. Actually, let $v\in C(w)$, $v'$ be a vertex on the shortest path from $w$ to $v$, then $v'\in C(w)$, since
	\begin{align*}
	\delta(v',w)=&~\delta(v,w)-\delta(v,v')\\
	<&~\delta(v,A_{i+1})-\delta(v,v')\\
	\le&~\delta(v',A_{i+1}).
	\end{align*}
	
	For every vertex $v\in V$, we also define a \emph{bunch} $B(v)$ as follows. For $w\in A_i\setminus A_{i+1}$, if $\delta(v,w)<\delta(v,A_{i+1})$, then $w$ is in the bunch $B(v)$. (For any $v$ and $i$, let the vertex in $A_i$ closest to $v$ be $u$, then $u$ must be in $A_j\setminus A_{j+1}$ for some $j\geq i$, so $u\in B(v)$.) It is easy to check that
	\[B(v)=\{w\in S:v\in C(w)\}.\]
	
	Now we derandomize the construction of $A_i$ and justify \cref{def:sr-tree-cover} \ref{item:treecover-c}). (That is, every vertex is in $\le kn^{1/k}(\ln n+1)$ trees.) For every vertex $v\in V$, since $v$ is only in the trees rooted in $B(v)$, it suffices to prove that $|B(v)|\le kn^{1/k}(\ln n+1)$. Suppose we have constructed $A_i$ and want to construct $A_{i+1}$ now. For $v\in V$, let $N_{i+1}(v)$ be the set of the $n^{1/k}(\ln n+1)$ closest vertices to $v$ in $A_i$. By \cite[Lemma 3.6]{TZ05}, a hitting set $A_{i+1}$ of the family $\{N_{i+1}(v):v\in V\}$ can be found in polynomial time with $|A_{i+1}|\le n^{-1/k}|A_i|$, and this finishes the construction of $A_{i+1}$. For each vertex $v\in V$ and level $i$, since $A_{i+1}\cap N_{i+1}(v)\ne\varnothing$, we have that $|B(v)\cap (A_i\setminus A_{i+1})|\le n^{1/k}(\ln n+1)$. It follows that $|B(v)|\le kn^{1/k}(\ln n+1)$.
	
	It remains to justify \cref{def:sr-tree-cover} \ref{item:treecover-b}). That is, for every $u\in S$ and $v\in V$, there is some $w\in S$ such that $u,v \in V(T(w))$, and $\dep_{T(w)}(u) + \dep_{T(w)}(v) \le (2k-1)\delta(u, v)$.
	\begin{itemize}
		\item If $u\in B(v)$, then we can pick $w=u$, and $\dep_{T(w)}(u)+\dep_{T(w)}(v)=\delta(u,v)$.
		\item Otherwise, assume $u\in A_{i_0}\setminus A_{i_0+1}$, and let $w_1$ be the vertex in $A_{i_0+1}$ closest to $v$. Then $\delta(w_1,v)\le\delta(u,v)$, thus $\delta(w_1,u)\le 2\delta(u,v)$. We also have that $w_1\in B(v)$, \ie $v\in C(w_1)$. If $w_1\in B(u)$, then we can pick $w=w_1$, and $\dep_{T(w)}(u)+\dep_{T(w)}(v)\le 3\delta(u,v)$.
		\item Otherwise ($w_1\not\in B(u)$), assume $w_1\in A_{i_1}\setminus A_{i_1+1}$, and let $w_2$ be the vertex in $A_{i_1+1}$ closest to $u$. Then $\delta(w_2,u)\le\delta(w_1,u)\le 2\delta(u,v)$, thus $\delta(w_2,v)\le 3\delta(u,v)$. We also have that $w_2\in B(u)$, \ie $u\in C(w_2)$. If $w_2\in B(v)$, then we can pick $w=w_2$, and $\dep_{T(w)}(u)+\dep_{T(w)}(v)\le 5\delta(u,v)$.
		\item Otherwise ($w_2\not\in B(v)$) ...
		\item Repeat this procedure until we find a tree $T(w)$ containing both $u$ and $v$.
	\end{itemize}
	The levels $i_0,i_1,\dots$ are strictly increasing, so we reach level $k-1$ in $O(k)$ time (if we did not terminate before). For every $v\in V$, we have $A_{k-1}\subseteq B(v)$, so the procedure indeed terminates in $O(k)$ time. It is easy to see that the stretch is at most $2k-1$.
	
	\section{Additional Figures and Tables}\label{apd:tables}
	\begin{table}[H]
	\centering
	\begin{tabular}[b]{|l|p{6.5cm}|p{6cm}|}
	\hline
	Notation & Meaning & Remarks\\
	\hline
	$\circ$ & path/sequence concatenation operator & \\
	\hline
	$w_H(u,v)$ & the weight of edge $(u,v)$ in $H$ & \multirow{4}{*}{\begin{minipage}{6cm}We omit $H$ when $H=G$ is the input graph.\end{minipage}}\\
	\cline{1-2}
	$\delta_H(u,v)$ & the distance between $u$ and $v$ in $H$ & \\
	\cline{1-2}
	$\pi_H(u,v)$ & the shortest $u$-$v$ path in $H$ & \\
	\cline{1-2}
	$\delta_H(u,S)$ & $\min\{\delta_H(u,v):v\in S\}$ & \\
	\hline
	$G[S]$ & the subgraph of $G$ induced by $S$ & $S\subseteq V$.\\
	\hline
	$P[u,v]$ & the portion between $u$ and $v$ of path $P$ & \multirow{4}{*}{\begin{minipage}{6cm}Assume $P=(x_0,x_1,\dots,x_{\ell-1},x_{\ell})$ where $x_0=u,x_{\ell}=v$; These notations sometimes emphasize the \emph{direction} of the path.\end{minipage}}\\
	$P[u,v)$ & $P[u,x_{\ell-1}]$ &\\
	$P(u,v]$ & $P[x_1,v]$ &\\
	$P(u,v)$ & $P[x_1,x_{\ell-1}]$ &\\
	\hline
	$l(v)$ & the level of $v$, or the largest $i$ such that $v\in U_i$ & \\
	\hline
	$G_{\ell}$ & the subgraph of $G$ induced by vertices with level $\le \ell$ & \\
	\hline
	$\caP_{\ell}(x,y)$ & the $x$-$y$ path in $\caT_{U_{\ell+1}}(U_{\ell})$ guaranteed by \cref{cor:T-ell(xy)} & \multirow{2}{*}{\begin{minipage}{6cm}$x\in U_{\ell}\setminus U_{\ell+1},y\in V\setminus U_{\ell+1}$.\end{minipage}}\\
	\cline{1-2}
	$T_{\ell}(x,y)$ & the tree in $\caT_{U_{\ell+1}}(U_{\ell})$ that contains $\caP_{\ell}(x,y)$ & \\
	\hline
	\end{tabular}
	\caption{Notation in this paper}\label{fig:def}
	\end{table}

	\begin{sidewaystable}
		\centering
		\begin{tabular}[b]{|c|c|c|c|c|c|c|c|c|c|c|c|c|c|c|}
			\hline
			failure & \# fault & size & query time & stretch & ref & remarks\\
			\hline
			edge & $1$ & $O(n^2\log n)$ & $O(\log n)$ & $1$ & \cite{DT02} & directed\\
			\hline
			edge & $1$ & $O(n^2\log n)$ & $O(1)$ & $1$ & \cite{CR02} & directed\\
			\hline
			vertex & $1$ & $O(n^2\log n)$ & $O(1)$ & $1$ & \cite{DTCR08, BK08, BK09} & directed\\
			\hline
			vertex & $1$ & $O(n^2)$ & $O(1)$ & $1$ & \cite{DZ17} & directed\\
			\hline
			vertex & $1$ & $O(k^5\epsilon^{-4}n^{1+1/k}\log^3 n)$ & $O(k)$ & $(2k-1)(1+\epsilon)$ & \cite{BK13} & unweighted\\
			\hline
			vertex & $2$ & $O(n^2\log^3 n)$ & $O(\log n)$ & $1$ & \cite{DP09} & directed\\
			\hline
			vertex & $2$ & $O(n^2)$ & $O(1)$ & reachability & \cite{Cho16} & directed\\
			\hline
			edge & $d$ & $O(m)$ & $O(d\log^{2.5} n\log\log n)$ & connectivity & \cite{PT07} & \\
			\hline
			edge & $d$ & $O(m\log\log n)$ & $O(d^2\log\log n)$ & connectivity & \cite{DP10} & \\
			\hline
			edge & $d$ & $O(m)$ & $O(d^2\log^{\epsilon}n)$ & connectivity & \cite{DP10} & \\
			\hline
			edge & $d$ & $O(n\log^2 n)$ & $O(d\log d\log\log n)$ & connectivity & \cite{DP17} & \\
			\hline
			edge & $d$ & $O\mleft(dkn^{1+1/k}\log(nW)\mright)$ & $O\mleft(d\log^2 n\log\log n\log\log (nW)\mright)$ & $(8k-2)(d+1)$ & \cite{CLPR12} &\\
			\hline
			edge & $d$ & $O(dn^2\log^2 n)$ & $O(d^2\log^2 n)$ & $2d+1$ & \cite{BGLP16} & \\
			\hline
			edge & $d$ & $O\mleft(n^3(\log n/\epsilon)^d(\log W/\log n)\mright)$ & $O(d^4\log\log W)$ & $1+\epsilon$ & \cite{CCFK17} & \\
			\hline
			edge & $d$ & $O\mleft(n^2(\log n/\epsilon)^d\cdot d\log W\mright)$ & $O(d^5\log n\log\log W)$ & $1+\epsilon$ & \cite{CCFK17} & \\
			\hline
			vertex & $d$ & $O\mleft(d^{1-2/c}mn^{1/c-1/(c\log(2d))}\log^2 n\mright)$ & $O\mleft(d^{2c+4}\log^2 n\log\log n\mright)$ & connectivity & \cite{DP10} & $c\ge 1$\\
			\hline
			vertex & $d$ & $O(m\log^6 n)$ & $O\mleft(d^2\log d\log^2 n\log\log n\mright)$ & connectivity & \cite{DP17} & \\
			\hline
			vertex & $d$ & $O(n/r)^{d+1}\frac{\sqrt{ndr}}{d!}+O(n\log^2 n)$ & $O(d\sqrt{r}\log^2 n)$ & $1$ & \cite{CMT19} & planar; $r\le n/d$\\
			\hline
			edge & $d$ & $O(Wn^{2+\mu}\log n)$ & $\tilde{O}(Wn^{2-\mu}d^2+Wnd^{\omega})$ & $1$ & \cite{BS19} & directed; $\mu\in[0,1]$\\
			\hline
			edge & $d$ & $O(n^2\log n)$ & $O(d^{\omega})$ & reachability & \cite{BS19} & directed\\
			\hline
		\end{tabular}
		\caption[]{previous results\tablefootnote{Unless stated in ``remark'' column, all data structures work on weighted undirected graphs.}}\label{fig:other_results}
		
		\ 
		
		\begin{tabular}[b]{|c|c|c|c|c|c|c|c|c|c|c|c|c|c|c|}
			\hline
			failure & \# fault & size & query time & stretch & ref & remarks\\
			\hline
			vertex & $d$ & $n^{3+1/c}\cdot O\mleft(\epsilon^{-1}\frac{\log^2 n\log(nW)}{\log d}\mright)^d$ & $O\mleft(\frac{d^{2c+6}\log^{10}n}{\log^2d}\mright)$ & $1+\epsilon$ & this paper & $c\ge 1$\\
			\hline
			vertex & $d$ & $n^{3+1/c}\frac{\log W}{\log n}\cdot O\mleft(\epsilon^{-1}\frac{\log^3 n}{\log d}\mright)^d$ & $O\mleft(\frac{d^{2c+6}\log^{10}n\log\log W}{\log^2d}\mright)$ & $1+\epsilon$ & this paper & $c\ge 1$\\
			\hline
			vertex & $d$ & $n^{2+1/c}\frac{\log d}{\log n}\cdot O\mleft(\epsilon^{-1}\frac{\log^2 n\log(nW)}{\log d}\mright)^{d+1}$ & $O\mleft(\epsilon^{-1}\frac{d^{2c+6}\log^{11}n\log(nW)}{\log^2d}\mright)$ & $1+\epsilon$ & this paper & $c\ge 1$\\
			\hline
			vertex & $d$ & $n^{2+1/c}\frac{\log W\log d}{\log^2 n}\cdot O\mleft(\epsilon^{-1}\frac{\log^3n}{\log d}\mright)^{d+1}$ & $O\mleft(\epsilon^{-1}\frac{d^{2c+6}\log^{12}n\log\log W}{\log^2d}\mright)$ & $1+\epsilon$ & this paper & $c\ge 1$\\
			\hline
            vertex & $d$ & $O\mleft(\frac{n^{2+1/c}d^3\log^{16}n\log(nW)}{\log^5d}\mright)$ & $O\mleft(\frac{d^{2c+9}\alpha^2(n)\log^{20}n\log\log(nW)}{\log^5d}\mright)$ & $O\mleft(\frac{d^{c+2}\log^6n}{\log d}\mright)$ & this paper & $c\ge 1$\\
			\hline
		\end{tabular}
		\caption{our results}\label{fig:our-results}
	\end{sidewaystable}
\end{document}